\documentclass[12pt]{amsart}

\usepackage[margin = 0.7in]{geometry} 
\usepackage{physics}
\usepackage{amssymb}
\usepackage{amsmath}
\usepackage{mathtools}

\title{Geometric derivation of the finite $N$ master loop equation}
\author{Omar Abdelghani, Ron Nissim}
\date{} 
\usepackage{physics}
\usepackage{cancel}
\usepackage{color}
\usepackage{float}
\usepackage{mathrsfs}
\usepackage{bm}
\usepackage{hyperref}
\usepackage{cleveref}
\usepackage{xcolor}
\hypersetup{
    colorlinks,
    linkcolor={red!50!black},
    citecolor={blue},
    urlcolor={blue}
}

\renewcommand{\div}{\text{ div}}

\renewcommand{\grad}{\nabla}

\newtheorem{theorem}{Theorem}[section]
\newtheorem{definition}{Definition}[section]
\newtheorem{lemma}{Lemma}[section]

\newtheorem{remark}{Remark}[section]
\newcommand{\Addresses}{{
  \bigskip
  \footnotesize

  Omar Abdelghani, \textsc{Department of Mathematics, Massachusetts Institute of Technology, Cambridge, Massachusetts 02139, USA.}\par\nopagebreak
  \textit{E-mail address}: \texttt{oabdel@mit.edu}
  
  \medskip
  
  Ron Nissim, \textsc{Department of Mathematics, Massachusetts Institute of Technology, Cambridge, Massachusetts 02139, USA.}\par\nopagebreak
  \textit{E-mail address}: \texttt{rnissim@mit.edu} 

}}

\begin{document}
\begin{abstract}

In this paper we provide a geometric derivation of the master loop equation for the lattice Yang-Mills model with structure group $G \in \{SO(N),SU(N), U(N)\}$. This approach is based on integration by parts on $G$. In the appendix we compare our approach to that of \cite{Ch19a} and \cite{J16} based on Schwinger-Dyson equations, and  \cite{SheSmZh22} based on stochastic analysis. In particular these approaches are all easily seen to be equivalent. The novelty in our approach is the use of intrinsic geometry of $G$ which we believe simplifies the derivation.
\end{abstract}
\maketitle
\tableofcontents 
\section{Introduction}

The study of a Quantum Yang-Mills theory is of great importance for understanding the standard model of particle physics. Although the physically relevant theory takes place in Minkowski space, standard arguments indicate that it is sufficient to study the Euclidean/probabilistic theory. This paper will only deal with the Euclidean theory, for a general survey paper see \cite{Ch18}.

Specifically, we will be dealing with the the lattice Yang-Mills model with structure group $G \in \{SO(N),SU(N), U(N)\}$. One approach to 'solving' the lattice Yang-Mills model is through computing Wilson loop expectations.  Wilson loop expectations are known to satisfy an equation known as a master loop equation.

Master loop equations for lattice gauge theories were first published by Makeenko and Migdal \cite{MaMi79} in 1979. However this derivation assumed a 'factorization property' without proof. In recent years there have been several rigorous derivations of master loop equations for classical groups such as $SO(N),SU(N),U(N)$. For instance see \cite{Ch19a}, \cite{J16}, \cite{SheSmZh22}, \cite{PaPfShYu23}, and \cite{CaPaSh23}.

The purpose of this paper is to provide a geometric derivation of the master loop equation for the lattice Yang-Mills model with structure group $G \in \{SO(N),SU(N), U(N)\}$. This approach is based on integration by parts on $G$, and is inspired by the derivations in \cite{Ch19a} and \cite{J16}. In particular, the approaches of \cite{Ch19a} and \cite{J16} rely on Schwinger-Dyson equations which they obtain from Stein's idea of exchangeable pairs. As remarked in a talk of Chatterjee \cite{Ch19b} (pointed out by Thierry Lévy), these Schwinger-Dyson equations are simply integration by parts on $G$ written in extrinsic coordinates. We show the equivalence of integration by parts and the Schwinger-Dyson equation in Appendix \ref{IBP section} and Appendix \ref{extrinsic derivation}, and generalize this to any compact Lie group with Riemannian structure inducing the Haar measure. The main sections of the paper are devoted to deriving the master loop equation for $G \in \{SO(N),SU(N), U(N)\}$ directly from integration by parts, relying on the intrinsic geometry of $G$ in order to simplify the calculations of Chatterjee and Jafarov. Lastly in Appendix \ref{Langevin section} we compare our approach to the stochastic analyis approach in \cite{SheSmZh22}.

\section*{Acknowledgements}
We thank Scott Sheffield and Sky Cao for many useful discussions.

\section{Definitions and notation}
\subsection{Lattice gauge theory}
In the following, $\Lambda$ is a finite two dimensional cell complex. 
\begin{definition} 
  A path $\gamma$ is a sequence $e_1,\dots ,e_n$ of oriented $1$-cells such that their union is connected in $\Lambda$.
  The set of all paths forms a groupoid $\mathscr{P}(\Lambda)$, where two paths $\gamma_1$ and $\gamma_2$ may be concatenated if 
  $$\partial_0 \gamma_2 = \partial_1\gamma_1(1)$$
\end{definition}
The idea of a lattice gauge theory is to discretize a notion of a connection on a principal $G$-bundle. By restricting paths to lie in a discrete set of $1$-cells, we may identify a connection with its finite set of parallel transport maps. This leads to the following: 
\begin{definition}
  Let $\Lambda$ be as before. A $G$-connection or $G$-gauge field is a homomorphism $$Q:\mathscr{P}(\Lambda)\to G$$.
  In other words, if $\gamma$ is a path from $a$ to $b$, and $\psi$ is a path from $b$ to $c$, then 
  $$Q(\psi \ast \gamma) = Q(\psi)Q(\gamma)$$
  And $$Q(\psi^{-1}) = Q^{-1}(\psi)$$

  It's clear from the definition that such a homomorphism is determined by its values on an oriented edge. Pick an arbitrary orientation for each edge. Let the set of such edges be denoted $E^+_{\Lambda}$. Then we have 
  $$\text{Hom}(\mathscr{P}(\Lambda), G)\cong G^{E^+_\Lambda}$$
  This characterization will be useful in defining lattice Yang-Mills measures.
\end{definition}
\begin{definition}
A path $\ell\in \mathscr{P}(\Lambda)$ is called a loop if its image has empty boundary.
\end{definition}
The most important observables in a lattice gauge theory are the Wilson loops. 
\begin{definition}
Let $\chi$ be an irreducible character on $G$, and let $\ell\in\mathcal{P}(\Lambda)$ be a loop. A Wilson loop is a functional of the form
$$W^\chi_\ell = \chi(Q(\ell))$$
\end{definition}

\begin{remark}
    For the rest of the paper we will only consider Wilson loops defined with the character $\chi(Q)=\mathrm{Tr}(Q)$ for $G=SO(N)$, and $\chi(Q)=\mathrm{Re}\mathrm{Tr}(Q)$ for $G\in\{SU(N),U(N)\}$. The proofs do not immediately generalize to other irreducible characters.
\end{remark}
\begin{definition}
Let $f$ be any class function (for our purposes usually an irreducible character). A plaquette is a $2$-cell in $\Lambda$. Let the set of faces be denoted $\mathcal{P}(\Lambda)$. Pick an arbitrary orientation, and let $\mathcal{P}^+_\Lambda$ denote the set of positively oriented plaquettes.  The lattice Yang-Mills measure with 't Hooft coupling is the probability measure 
$$Z_{\beta, \Lambda, N}^{-1}\exp(\beta N \sum_{p\in \mathcal{P}^+_\Lambda} f(Q(\partial p)))\prod_{e\in E_\Lambda}dg_e$$
Where $dg$ is the Haar measure.
\end{definition}

\subsection{Path operations}
The following is a list of operations on loops that will be relevant in defining the master loop equations. The terminology and notation is based on that of \cite{Ch19a}. The master loop equation is based around integration by parts on an edge $e$ in $\Lambda$. Thus, for every path, We define a set $C_\ell$ indexing the occurrences of $e^{\pm}$ in $\ell$. Moreover, if $x\in C_\ell$, then $\omega_x\in \qty{-1,1}$ is the orientation of that instance of $e$.
\begin{definition}
Let $\ell= a_1 e^{\omega_1}\dots a_n e^{\omega_n} a_{n+1}$ be a loop. We define 
$$\ell\setminus e_x = a_{x+1}\dots e^{\omega_n} a_{n+1}\dots a_{x}$$
In other words, this is the string formed by excising $e_x$, ordered starting from the edge after $e_x$.
\end{definition}
\begin{definition}[Positive merger]
Let $\ell_1$ and $\ell_2$ be loops. 
$$\ell_1\oplus_{x,y}\ell_2 = (\ell_1\setminus e_x) e^{\omega_x} (\ell_2\setminus e_y e)^{\omega_x \omega_y}$$
\end{definition}
\begin{definition}[Negative merger]
Let $\ell_1$ and $\ell_2$ be loops. 
$$\ell_1\oplus_{x,y}\ell_2 = (\ell_1\setminus e_x)(\ell_2\setminus e_y)^{-\omega_x \omega_y}$$
\end{definition}
\begin{definition}[Positive split]
Let $\ell = a_1 e^{\omega_1}\dots a_n e^{\omega_n}a_{n+1}$ be a loop such that $\abs{C_\ell}>1$. If $x\neq y\in C_\ell$ and $\omega_x \omega_y = 1$, the positive split at $x,y$ is the pair of loops 
\begin{align*}
\times^1_{x,y}\ell &= a_{x+1}e^{\omega_{x+1}}\dots a_{y-1}e^{\omega_{y-1}}e^{\omega_y}\\
\times^2_{x,y}\ell &= a_{y+1}e^{\omega_{y+1}}\dots a_{x-1}e^{\omega_{x-1}}e^{\omega_x}
\end{align*}

\end{definition}
\begin{definition}[Negative split]
Let $\ell = a_1 e^{\omega_1}\dots a_n e^{\omega_n}a_{n+1}$ be a loop such that $\abs{C_\ell}>1$. If $x\neq y\in C_\ell$ and $\omega_x \omega_y = -1$, the positive split at $x,y$ is the pair of loops 
\begin{align*}
\times^1_{x,y}\ell &= a_{x+1}e^{\omega_{x+1}}\dots a_{y-1}e^{\omega_{y-1}}\\
\times^2_{x,y}\ell &= a_{y+1}e^{\omega_{y+1}}\dots a_{x-1}e^{\omega_{x-1}}
\end{align*}
\end{definition}
\begin{definition}[Positive twist]
Let $\ell$ be a loop such that $\abs{C_\ell}>1$. WLOG suppose $e_x\in \times^1_{x,y}\ell$ and $e_y\in \times^2_{x,y}$. If $x\neq y\in C_\ell$ and $\omega_x\omega_y = -1$, the positive twist of $\ell$ at $x,y$ is the loop 
$$\propto_{x,y}\ell =\qty(\times^1_{x,y}\ell)e^{\omega_x} \qty(\times^2_{x,y}\ell)^{-1}e^{\omega_y} $$
\end{definition}
\begin{definition}[Negative twist]
Let $\ell$ be a loop such that $\abs{C_\ell}>1$. Let $x\neq y\in C_\ell$ and $\omega_x\omega_y = 1$.  WLOG let $e_x\in \times^1)_{x,y}\ell$. Then the negative twist is 
$$\propto_{x,y} \ell = \times^1_{x,y} \ell \ominus_{x,y} \qty(\times^2_{x,y}\ell)^{-1}$$
\end{definition}

\subsection{$SO(N)$ definitions and conventions}
Recall the following: $$SO(N) = \qty{g\in GL_N(\mathbb{R})| g^T g = I}$$

The tangent space at $g$ is the space of matrices $$T_g SO(N) = g \mathfrak{so}(N) = \qty{M\in \mathbb{R}^{N^2}| gX +X^T g^{-1} = 0}$$

There is a natural bi-invariant metric on $SO(N)$:
\begin{definition}
Let $X,Y\in T_g SO(N)$. Then $$\ev{X,Y} = \frac{1}{2}\Tr(X^TY)$$
(Note that this is the restriction of half the Euclidean metric to $SO(N)$)
\end{definition}
With respect to this metric, 
$$X_{ij} = g(e_i e_j^T - e_j e_i^T)$$
is an orthonormal frame on $TG$. 
\newline 
For $SO(N)$ lattice gauge theory, we use the lattice Yang-Mills measure defined by the character $f(Q) = \Tr Q$.
\subsection{$SU(N)$ and $U(N)$ definitions and conventions}
Recall that $$SU(N) = \qty{g\in GL_N(\mathbb{C})| g^\dagger g = I, \det g = 1}$$
$$U(N) = \qty{g\in GL_N(\mathbb{C})| g^\dagger g = I}$$

The tangent space at $g$ is the space of matrices 
$$T_g SU(N) = g \mathfrak{su}(N) = \qty{M\in \mathbb{R}^{N^2}| gX +X^\dagger g^{-1}=0,\Tr(g^{-1}X)=0}$$
$$T_g U(N) = g \mathfrak{su}(N) = \qty{M\in \mathbb{R}^{N^2}| gX +X^\dagger g^{-1}=0}$$

There is a natural bi-invariant metric on $SU(N)$:
\begin{definition}
Let $X,Y\in T_g SU(N)$ or $T_gU(N)$. Then $$\ev{X,Y} = \frac{1}{2}\Re\Tr(X^\dagger Y)$$
\end{definition}
(Note that this is the restriction of half the Euclidean metric to $SU(N)$)
\newline 
For $SU(N)$ lattice gauge theory, we use the lattice Yang-mills measure defined by the class function $f(Q)  = \Re \Tr Q$
\newline
The $SU(N)$ case is complicated in that the character associated with the defining representation is now complex valued. Thus we need to make sense of gradients of functions $f\in C^\infty(SU(N), \mathbb{C})$
\begin{definition}
Let $f\in C^\infty(SU(N), \mathbb{C})$. Define $\grad f\in \mathfrak{X}(SU(N))\otimes \mathbb{C}$ by 
$$\grad f = \grad \Re f + i \grad \Im f$$
In other words, we extend the gradient in the natural way to a complex-linear map of smooth functions.
\end{definition}
Similarly, we extend the metric. 
\begin{definition}
Let $X, Y\in \mathfrak{X}(SU(N))\otimes \mathbb{C}$. Then 
$$\ev{X,Y} = \ev{\Re X, \Re Y} - \ev{\Im X, \Im Y} + i\qty(\ev{\Re X, \Im Y} + \ev{\Im X, \Re Y})$$
\end{definition}
With this definition in mind, we can formulate the basis of the proof of the main theorem: 
\begin{lemma}[Laplacian integration by parts]
Let $f, g\in C^\infty(G;\mathbb{C})$, with $G$ a compact Lie group. Equip $G$ with a bi-invariant metric. With respect to this metric, 
$$\int_G g\Delta f = -\int_G \ev{\grad f, \grad g}$$
\end{lemma}
In the sequel we will show that when $f$ and $g$ are Wilson loop correlation functions, laplacian integration by parts directly reduces to the master loop equation.
\section{Main Results}
The first main result is the master loop equation for $SO(N)$
\begin{theorem}[$SO(N)$ master loop equation]\label{thm1}
  Let $(\ell_1\dots \ell_n)$ be a sequence of loops. Let $e$ be an edge that lies in at least one of $\ell_i$. Let $\mathbb{E}$ denote expectations with respect to the $SO(N)$ lattice yang mills measure. Let $\mathcal{P}^+(e)$ denote the set of positively oriented plaquettes containing $e$. Finally, let $A_i$ be the set of occurrences of the edge $e\in E_\Lambda^+$ in $\ell_i$, $B_i$ the set of occurrences of $e^{-1}$, and $C_i = A_i \cup B_i$. Then 
 \begin{align}
  \begin{split}
  &(N-1)m \mathbb{E}[W_{\ell_1}\dots W_{\ell_n}] = \sum_{x\neq y\in C_1, \omega_x\omega_y = 1}\mathbb{E}[W_{\propto_{x,y}\ell_1}W_{\ell_2}\dots W_{\ell_n}] \\
  &- \sum_{x\neq y\in C_1, \omega_x \omega_y = -1} \mathbb{E}[W_{\propto_{x,y}\ell_1}W_{\ell_2}\dots W_{\ell_n}] +\sum_{x,y\in C_1, \omega_x\omega_y = -1}\mathbb{E}[W_{\times^1_{x,y}\ell_1}W_{\times^2_{x,y}\ell_1}W_{\ell_2}\dots W_{\ell_n}]\\
  &- \sum_{x\neq y\in C_1, \omega_x \omega_y = 1}\mathbb{E}[W_{\times^1_{x,y}\ell_1}W_{\times^2_{x,y} \ell_1}W_{\ell_2}\dots W_{\ell_n}] +\sum_{i=2}^n \sum_{x\in C_1, y\in C_i} \mathbb{E}\qty[W_{\ell_1\ominus_{x,y}\ell_i}\prod_{j\neq 1,i}W_{\ell_j}]\\
  &- \sum_{i=2}^n \sum_{x\in C_1, y\in C_i} \mathbb{E}\qty[W_{\ell_1\oplus_{x,y}\ell_i} \prod_{j\neq i, 1 } W_{\ell_j}] +N\beta \sum_{p\in\mathcal{P}^+(e)} \sum_{x\in C_1}\mathbb{E}[W_{\ell_1\ominus_x p} W_{\ell_2}\dots W_{\ell_n}]\\
  &- N\beta \sum_{p\in\mathcal{P}^+(e)}\sum_{x\in C_1}\mathbb{E}[W_{\ell_1\oplus_x p}W_{\ell_2}\dots W_{\ell_n}]
  \end{split}
 \end{align}
\end{theorem}
\begin{theorem}[$SU(N)$ and $U(N)$ master loop equation]\label{thm2}
Let $(\ell_1, \dots, \ell_n)$ be a sequence of loops. Let $e$ be an edge that lies in at least one of $\ell_i$. Let $\eta=0$ for $G=U(N)$ and $\eta=1$  when $G=SU(N)$. Let $A_i$, $B_i$, $C_i$, $\mathcal{P}^+(e)$ be as before. In addition, let $\mathcal{P}(e)$ be the set of plaquettes containing $e$. Finally, let $t_i = \abs{A_i}-\abs{B_i}$ and $t= \sum_i t_i$. Then 
\begin{align*}
&\qty(mN - \frac{\eta t_1 t}{N})\mathbb{E}\qty[W_{\ell_1}\dots W_{\ell_n}]= \sum_{x,y\in C_1, \omega_x\omega_y = -1}\mathbb{E}\qty[W_{\times^1_{x,y}\ell_1}W_{\times^2_{x,y}\ell_1}W_{\ell_2}\dots W_{\ell_n}] \\ &- \sum_{x\neq y \in C_1, \omega_x \omega_y = 1}\mathbb{E}\qty[W_{\times^1_{x,y}\ell_1}W_{\times^2_{x,y}\ell_1}W_{\ell_2}\dots W_{\ell_n}] +\sum_{i=2}^n \sum_{x\in C_1, y\in C_i, \omega_x \omega_y = -1}\mathbb{E}\qty[W_{\ell_1\ominus_{x,y}\ell_i}\prod_{j\neq i, 1}W_{\ell_j}] \\ &- \sum_{i=2}^n \sum_{x\in C_1, y\in C_i, \omega_x \omega_y = 1} \mathbb{E}\qty[W_{\ell_1\oplus_{x,y}\ell_i}\prod_{j\neq i, 1}W_{\ell_j}]
+ \frac{\beta N}{2}\sum_{i=2}^n \sum_{p\in\mathcal{P}^+(e), x\in C_1}\mathbb{E}\qty[W_{\ell_1\ominus_x p} W_{\ell_2}\dots W_{\ell_n}] 
\\&-\frac{\beta N}{2}\sum_{i=2}^n \sum_{p\in\mathcal{P}^+(e), x\in C_1}\mathbb{E}\qty[W_{\ell_1\oplus_x p} W_{\ell_2}\dots W_{\ell_n}] -\eta \beta \sum_{p\in\mathcal{P}(e)}t_1t_p \mathbb{E}[W_{\ell_1}W_{p^{-1}}W_{\ell_2}\dots W_{\ell_n}]\\
\end{align*}
\end{theorem}

\section{$SO(N)$ analysis}
The goal of this section is to prove \cref{thm1}

\subsection{Gradient identities for Wilson loops}
\begin{lemma}\label{lem1}
Let $\ell$ be a Wilson loop. Let $C$ denote the set of occurrences of the edge $e$ in $\ell$. Let $\grad_e$ denote the gradient with respect to the edge $e$. Then 
\begin{equation}
  \grad_e W_{\ell} = \sum_{x\in C} Q_{\ell\setminus e_x}^{-\omega_x} - Q_e Q_{\ell\setminus e_x}^{\omega_x}Q_e
\end{equation}

\end{lemma}
\begin{proof}
  $W_{\ell}$ extends to a smooth function in a neighborhood of $SO(N)$ in the obvious way. We may thus first compute the euclidean differential:
  $$d_e W_{\ell} = \sum_{x\in C} d_x W_{\ell}$$
  For $d_x$,  we may rewrite $W_{\ell} = \Tr(Q_{\ell\setminus e_x} g^{\omega_x})$. By orthogonality, 
  $$ =  \Tr(Q_{\ell\setminus e_x}^{\omega_x}g)$$
  Thus, the euclidean differential is 
  $$d_e W_\ell(H)=\sum_{x\in C} \Tr(Q_{\ell\setminus e_x}^{\omega_x} H)$$
  Recall that the gradient is defined by the identity $\ev{\grad f, H} =\frac{1}{2}\Tr(\grad f^T H)=df(H)$.
  Therefore, the euclidean gradient is 
  $$\grad^{euclidean}_e W_{\ell}(g) = \sum_{x\in C} 2Q_{\ell\setminus e_x}^{-\omega_x}$$
  Recall that the tangent projection onto $SO(N)$ is 
  $$P_g(X) = \frac{1}{2}X - \frac{1}{2}gX^T g$$
  Setting $g=Q_e$, we finally arrive at the result: 
  $$\grad_e W_{\ell} = \sum_{x\in C} Q_{\ell\setminus e_x}^{-\omega_x} - Q_e Q_{\ell\setminus e_x}^{\omega_x}Q_e$$
\end{proof}
\begin{lemma}\label{lem2}
  Let $\ell_1$ and $\ell_2$ be Wilson loops. 
  \begin{equation}
    \ev{\grad W_{\ell_1}, \grad W_{\ell_2}} = \sum_{x\in C_1, y\in C_2} W_{\ell_1\ominus_{x,y}\ell_2} - \sum_{x\in C_1, y\in C_2} W_{\ell_1\oplus_{x,y}\ell_2}
  \end{equation}
\end{lemma}
\begin{proof}
  \begin{align*} 
    &\ev{\grad W_{\ell_1}, \grad W_{\ell_2}} = \sum_{x\in C_1, y\in C_2}\ev{Q_{\ell_1\setminus e_x}^{-\omega_x} - g Q_{\ell_1\setminus e_x}^{\omega_x}g, Q_{\ell_2\setminus e_y}^{-\omega_y} - g Q_{\ell_2\setminus e_y}^{\omega_y}g}\\
    &= \sum_{x\in C_1, y\in C_2} \frac{1}{2}\qty(\Tr(Q_{\ell_1\setminus e_x}^{\omega_x} Q_{\ell_2\setminus e_y}^{-\omega_y}) + \Tr(g^{-1} Q_{\ell_1\setminus e_x}^{-\omega_x} Q_{\ell_2\setminus e_y}^{\omega_y}g)) - \frac{1}{2}\qty(\Tr(g^{-1}Q_{\ell_1\setminus e_x}^{-\omega_x} g^{-1} Q_{\ell_2\setminus e_y}^{-\omega_y}) + \Tr(Q_{\ell_1\setminus e_x}^{\omega_x} g Q_{\ell_2\setminus e_y}^{\omega_y}g))\\
    &= \sum_{x\in C_1, y\in C_2} \Tr(Q_{\ell_1\setminus e_x}^{\omega_x}Q_{\ell_2\setminus e_y}^{-\omega_y}) - \Tr(Q_{\ell_1\setminus e_x}^{\omega_x }g Q_{\ell_2\setminus e_y}^{\omega_y}g) = \sum_{x\in C_1, y\in C_2} W_{\ell_1\ominus_{x,y} \ell_2} - W_{\ell_1\oplus_{x,y}\ell_2}
  \end{align*}
\end{proof}
\subsection{Laplacian of Wilson loops}
\begin{lemma}\label{lem3}
  Let $P_g$ denote the tangent projection. Let $L_X$ and $R_X$ denote left and right multiplication, respectively. 
  Then 
  \begin{equation}
    \Tr(P_gL_X R_Y P_g)= \frac{1}{2}\Tr X \Tr Y - \frac{1}{2}\Tr(g^{-1}X^T g Y)
  \end{equation}
\end{lemma}
\begin{proof}
  \begin{align*}
    &\Tr(P_g L_X R_YP_g) = \sum_{i<j}\ev{ge_i e_j^T - ge_j e_i^T, X(ge_i e_j^T - g e_j e_i^T)Y}\\
  &=\frac{1}{2}\sum_{i<j}\qty[\Tr(e_j e_i^T g^{-1} X g e_i e_j^TY) - \Tr(e_i e_j^T g^{-1} X ge_i e_j^T Y) - \Tr(e_j e_i^T g^{-1} X g e_j e_i^T Y) + \Tr(e_i e_j^T g^{-1} X ge_j e_i^T Y)]\\
   &=\frac{1}{2}\sum_{i<j} \qty[(g^{-1}Xg)_{ii}Y_{jj} + (g^{-1}X g)_{jj}Y_{ii} - (g^{-1}Xg)_{ij} Y_{ij} - (g^{-1}Xg)_{ji}Y_{ji}] \\
   &=\frac{1}{2}\Tr(X)\Tr(Y) - \frac{1}{2}\Tr(g^{-1}X^T g Y)\\
  \end{align*}
\end{proof}
\begin{lemma}\label{lem4}
  Let $W_\ell$ be a Wilson loop and $\Delta_e$ the Laplace-Beltrami operator at the edge $e$. Let $m$ be the number of occurrences of $\pm e$ in $\ell$. Then 
  \begin{equation}
    \begin{split}
    \Delta_e W_{\ell} = -(N-1)mW_{\ell} -\sum_{x\neq y\in C, \omega_x \omega_y = 1} W_{\ell\times^1_{x,y}}W_{\ell\times^2_{x,y}} + \sum_{x\neq y \in C, \omega_x \omega_y = 1} W_{\ell\propto_{x,y}}\\
+ \sum_{x,y\in C, \omega_x\omega_y = -1} W_{\ell\times^1_{x,y}}W_{\ell\times^2_{x,y}} - \sum_{x,y\in C, \omega_x\omega_y = -1} W_{\ell\propto_{x,y}}
    \end{split}
  \end{equation}
\end{lemma}
\begin{proof}
Recall that for a vector field $X$, $\div X = \Tr \nabla X$ where $\nabla X$ is the covariant derivative of $X$. On a submanifold, 
the covariant derivative is the tangent projection of the ambient covariant derivative. Thus it suffices to compute 
$$\div X  = \Tr(P_{g} D X P_g)$$
Where $DX$ is the euclidean covariant derivative. Now, by \cref{lem1}:
$$D\grad_e W_{\ell} = \sum_{x\in C}DQ_{\ell\setminus e_x}^{-\omega_x} - \sum_{x\in C}D(g Q^{\omega_x}_{\ell\setminus e_x}g)$$
$$D\nabla_e W_{\ell}(H) = \sum_{x\neq y\in C} D_yQ^{-\omega_x}_{\ell\setminus e_x}(H) - \sum_{x\in C} \qty(H Q^{\omega_x}_{\ell\setminus e_x}g + g Q^{\omega_x}_{\ell\setminus e_x}H) - \sum_{x\neq y \in C} g D_yQ^{\omega_x}_{\ell\setminus e_x}(H)g $$
Now expanding further requires casework. In particular, in the first term, the $y$th occurrence of $e$ is $g$ if $\omega_x \omega_y = 1$ and is $g^{-1}$ otherwise. The opposite is true for the third term.
Recall that 
$$D(g\mapsto g^{-1})(H)  = -g^{-1}Hg^{-1}$$
Let $W_{\ell}(g_x\to H)$ denote the linear map formed by substituting $H$ for the $x$th occurrence of $e$. We thus have
\begin{equation}\label{eucderiv}
  \begin{split}
  D\nabla_e W_{\ell}(H)&= \sum_{x,y\in C, \omega_x\omega_y = -1} Q_{\ell\setminus e_x}^{-\omega_x}(g_y\mapsto H) - \sum_{x\neq y \in C, \omega_x\omega_y = 1} Q_{\ell\setminus e_x}^{-\omega_x}(g_{y}^{-1}\mapsto g^{-1}Hg^{-1})\\
  &-\sum_{x\in C}\qty(HQ^{\omega_x}_{\ell\setminus e_x}g + gQ^{\omega_x}_{\ell\setminus e_x}H) - \sum_{x\neq y\in C, \omega_x \omega_y = 1} g Q^{\omega_x}_{\ell\setminus e_x}(g_y\mapsto H)g \\&+ \sum_{x,y\in C, \omega_x\omega_y = -1}gQ^{\omega_x}_{\ell\setminus e_x}(g_y\mapsto g^{-1}Hg^{-1})g
  \end{split}
\end{equation}
We can now apply \cref{lem3}, as every term in this sum is an operator of the form $L_X R_Y$.
We can first consider 
$$\Tr(H\mapsto P_g Q^{-\omega_x}_{\ell\setminus e_x} (g_y \to H)P_g)$$
We can write $Q_{\ell\setminus e_x} = P_+ g^{\omega_y}P_-$. Then $Q_{\ell\setminus e_x}^{-\omega_x} = P_{-\omega_x}^{-\omega_x} g P^{-\omega_x}_{\omega_x}$ and so the trace is 
$$\frac{1}{2}\Tr(P_+)\Tr(P_-) - \frac{1}{2}\Tr(g^{-1}P_+^{-1} g P_-) = \frac{1}{2}W_{\times^1_{x,y} \ell}W_{\times^2_{x,y} \ell} - \frac{1}{2} W_{\propto_{x,y}\ell}$$

Similarly, for the case $\omega_x \omega_y = 1$, $Q_{\ell\setminus e_x}^{-\omega_x} = P^{-\omega_x}_{-\omega_x}g^{-1}P^{-\omega_x}_{\omega_x}$. So we compute the trace of $P^{-\omega_x}_{-\omega_x}g^{-1}Hg^{-1}P^{-\omega_x}_{\omega_x}$
Which equals 
$$\frac{1}{2}\Tr(P^{-\omega_x}_{-\omega_x}g^{-1})\Tr(g^{-1}P^{-\omega_x}_{\omega_x}) - \frac{1}{2} \Tr(P^{\omega_x}_{-\omega_x} P^{-\omega_x}_{\omega_x}) = \frac{1}{2}W_{\times^1_{x,y}\ell}W_{\times^2_{x,y}\ell} - \frac{1}{2}W_{\propto_{x,y}\ell}$$

Next we have $$\Tr(H\mapsto H Q^{\omega_x}_{\ell\setminus e_x} g + g Q^{\omega_x}_{\ell\setminus e_x}H)$$
By \cref{lem3}, 
\begin{align*}
&=\frac{N}{2}\Tr(Q^{\omega_x}_{\ell\setminus e_x}g) - \frac{1}{2}\Tr(Q^{\omega_x}_{\ell\setminus e_x}g) + \frac{N}{2}\Tr(Q^{\omega_x}_{\ell\setminus e_x}g) - \frac{1}{2}\Tr(Q^{\omega_x}_{\ell\setminus e_x}g)\\
&= (N-1) W_{\ell}
\end{align*}
We finally come to our last type of expression. 
$$\Tr(H\mapsto gQ^{\omega_x}_{\ell\setminus e_x} (g_y \mapsto H)g)$$
Writing $Q_{\ell\setminus e_x} = P_+ g^{\omega_y} P_-$, we have $gQ^{\omega_x}_{\ell\setminus e_x}g(H) = gP^{\omega_x}_{\omega_x} H P^{\omega_x}_{-\omega_x}g$
Thus, the trace is 
$$ = \frac{1}{2}\Tr(gP^{\omega_x}_{\omega_x})\Tr(P^{\omega_x}_{-\omega_x} g) - \frac{1}{2}\Tr(g^{-1} P^{-\omega_x}_{\omega_x} g^{-1}g P^{\omega_x}_{-\omega_x} g) = \frac{1}{2} W_{\times^1_{x,y}}W_{\times^2_{x,y}\ell} - \frac{1}{2}W_{\propto_{x,y}\ell}$$
For $\omega_x\omega_y = -1$, $gQ^{\omega_x}_{\ell\setminus e_x}g(H) =g P^{\omega_x}_{\omega_x} g^{-1}Hg^{-1}P^{\omega_x}_{-\omega_x}g $. Thus the final trace is 
$$ \frac{1}{2}\Tr(g P^{\omega_x}_{\omega_x}g^{-1})\Tr(g^{-1}P^{\omega_x}_{-\omega_x}g) - \frac{1}{2}\Tr(g^{-1}g P^{-\omega_x}_{\omega_x}g^{-1}g g^{-1}P^{\omega_x}_{-\omega_x} g)$$
$$ = \frac{1}{2}\Tr(P^{\omega_x}_{\omega_x})\Tr(P^{\omega_x}_{-\omega_x}) - \frac{1}{2}\Tr(P^{-\omega_x}_{\omega_x} g^{-1}P^{\omega_x}_{-\omega_x}g)= \frac{1}{2}W_{\times^1_{x,y}\ell}W_{\times^2_{x,y}\ell} - \frac{1}{2}W_{\propto_{x,y}\ell}$$
We can now insert these identities back into \cref{eucderiv}.
\begin{align*} 
  \Delta_e W_{\ell} &= \sum_{x,y\in C \omega_x\omega_y = -1}\qty(\frac{1}{2}W_{\times^1_{x,y} \ell}W_{\times^2_{x,y} \ell} - \frac{1}{2} W_{\propto_{x,y}\ell})- \sum_{x\neq y\in C, \omega_x \omega_y = 1} \qty( \frac{1}{2}W_{\times^1_{x,y}\ell}W_{\times^2_{x,y}\ell} - \frac{1}{2}W_{\propto_{x,y}\ell})\\
  &- \sum_{x\in C}(N-1)W_{\ell} - \sum_{x\neq y \in C, \omega_x \omega_y = 1}\qty( \frac{1}{2}W_{\times^1_{x,y}\ell}W_{\times^2_{x,y}\ell} - \frac{1}{2}W_{\propto_{x,y}\ell}) \\&+ \sum_{x,y\in C \omega_x \omega_y = -1}\qty( \frac{1}{2}W_{\times^1_{x,y}\ell}W_{\times^2_{x,y}\ell} - \frac{1}{2}W_{\propto_{x,y}\ell})\\
  &=-(N-1)m W_\ell +\sum_{x,y\in C, \omega_x \omega_y = -1}\qty(W_{\times^1_{x,y}\ell}W_{\times^2_{x,y}\ell} - W_{\propto_{x,y}\ell}) - \sum_{x\neq y \in C, \omega_x \omega_y = 1}\qty(W_{\times^1_{x,y}\ell}W_{\times^2_{x,y}\ell}-W_{\propto_{x,y}\ell}) 
\end{align*}
\end{proof}
\subsection{$SO(N)$ master loop equation}
\begin{proof}[Proof of \cref{thm1}.] Let $(\ell_1, \dots \ell_n)$ be a sequence of loops. We apply Laplace integration by parts to the Wilson loop correlation function:
\begin{align*} 
  &-\mathbb{E}[(\Delta_e W_{\ell_1})W_{\ell_2}\dots W_{\ell_n}]= Z^{-1}\int_{G^{E^+_\Lambda}} (\Delta_e W_{\ell_1})W_{\ell_2}\dots W_{\ell_n}\exp(\beta N \sum_{p\in\mathcal{P}^+_\Lambda} \Tr(Q_p))d\mu\\
  &= Z^{-1}\int_{G^{E^+_\Lambda}} \ev{\grad W_{\ell_1}, \grad\qty(W_{\ell_2}\dots W_{\ell_n}\exp(\beta N \sum_{p\in\mathcal{P}^+_\Lambda} \Tr(Q_p)))}d\mu\\
  &=\sum_{i = 2}^n \mathbb{E}\qty[\ev{\grad W_{\ell_1}, \grad W_{\ell_i}}\prod_{j\neq 1, i} W_{\ell_j}]+ \beta N \sum_{p\in \mathcal{P}^+_\Lambda}\mathbb{E}[\ev{\grad W_{\ell_1}, \grad W_{p}}W_{\ell_2}\dots W_{\ell_n}]\\ 
  &= \sum_{i=2}^{n}\sum_{x\in C_1, y\in C_i}\mathbb{E}\qty[W_{\ell_1\ominus_{x,y}\ell_i }\prod_{j\neq 1,i} W_{\ell_j}] -  \sum_{i=2}^{n}\sum_{x\in C_1, y\in C_i}\mathbb{E}\qty[W_{\ell_1\oplus_{x,y}\ell_i }\prod_{j\neq 1,i} W_{\ell_j}]\\
  &+\beta N  \sum_{i=2}^{n}\sum_{x\in C_1, y\in C_i}\mathbb{E}\qty[W_{\ell_1\ominus_{x,y}p }W_{\ell_2}\dots W_{\ell_n}] - \beta N  \sum_{i=2}^{n}\sum_{x\in C_1, y\in C_i}\mathbb{E}\qty[W_{\ell_1\oplus_{x,y}p }W_{\ell_2}\dots W_{\ell_n}]
\end{align*}
On the left hand side, we have 
\begin{align*}&-\mathbb{E}[(\Delta_e W_{\ell_1})W_{\ell_2}\dots W_{\ell_n}] \\&= (N-1)m \mathbb{E}[W_{\ell_1}\dots W_{\ell_n}] + \sum_{x\neq y \in C_1, \omega_x\omega_y = 1}\qty(\mathbb{E}[W_{\times^1_{x,y}\ell_1}W_{\times^2_{x,y}\ell_1}W_{\ell_2}\dots W_{\ell_n}] - \mathbb{E}[W_{\propto_{x,y}\ell_1}W_{\ell_2}\dots W_{\ell_n}])\\
  &- \sum_{x,y\in C_1, \omega_x \omega_y = -1}\qty(\mathbb{E}[W_{\times^1_{x,y}\ell_1}W_{\times^2_{x,y}\ell_1}W_{\ell_2}\dots W_{\ell_n}] - \mathbb{E}[W_{\propto_{x,y}\ell_1}W_{\ell_2}\dots W_{\ell_n}])
\end{align*}
Setting the two sides equal and rearranging terms gives the result: 
\begin{align*}
&(N-1)m\mathbb{E}[W_{\ell_1}\dots W_{\ell_n}] =\sum_{x,y\in C_1, \omega_x \omega_y = -1}\qty(\mathbb{E}[W_{\times^1_{x,y}\ell_1}W_{\times^2_{x,y}\ell_1}W_{\ell_2}\dots W_{\ell_n}] - \mathbb{E}[W_{\propto_{x,y}\ell_1}W_{\ell_2}\dots W_{\ell_n}]) \\
&- \sum_{x\neq y \in C_1, \omega_x\omega_y = 1}\qty(\mathbb{E}[W_{\times^1_{x,y}\ell_1}W_{\times^2_{x,y}\ell_1}W_{\ell_2}\dots W_{\ell_n}] - \mathbb{E}[W_{\propto_{x,y}\ell_1}W_{\ell_2}\dots W_{\ell_n}])\\
&= \sum_{i=2}^{n}\sum_{x\in C_1, y\in C_i}\mathbb{E}\qty[W_{\ell_1\ominus_{x,y}\ell_i }\prod_{j\neq 1,i} W_{\ell_j}] -  \sum_{i=2}^{n}\sum_{x\in C_1, y\in C_i}\mathbb{E}\qty[W_{\ell_1\oplus_{x,y}\ell_i }\prod_{j\neq 1,i} W_{\ell_j}]\\
  &+\beta N  \sum_{i=2}^{n}\sum_{x\in C_1, y\in C_i}\mathbb{E}\qty[W_{\ell_1\ominus_{x,y}p }W_{\ell_2}\dots W_{\ell_n}] - \beta N  \sum_{i=2}^{n}\sum_{x\in C_1, y\in C_i}\mathbb{E}\qty[W_{\ell_1\oplus_{x,y}p }W_{\ell_2}\dots W_{\ell_n}]
\end{align*}
This concludes the proof of the $SO(N)$ master loop equation.
\end{proof}
\section{$SU(N)$ and $U(N)$ analysis}
We now move on to the proof of \cref{thm2}

The procedure is mostly analogous. However unlike $SO(N)$, not every element of these groups is conjugate to its inverse. Thus, the orientations of Wilson loops will become more relevant in the analysis. This is reflected in the fact that Wilson loops are now complex valued. We will introduce a parameter $\eta$ that vanishes when $G=U(N)$ and is $1$ when $G=SU(N)$.
\subsection{Gradients of Wilson loops}
\begin{lemma}\label{gradlemma2} 
  Let $\ell$ be a  loop. 
 \begin{equation} \grad \Re W_\ell = 
  \sum_{x\in C_\ell} Q_{\ell\setminus e_x}^{-\omega_x} - Q_e Q_{\ell\setminus e_x}^{\omega_x}Q_e + \eta\frac{2i\omega_x}{N}\Im W_\ell Q_e
 \end{equation}
 
\end{lemma}
\begin{proof}
 As before, the differential is
  $$d\Re W_\ell(H) = \sum_{x\in C_\ell}d_x \Re W_\ell(H) = \sum_{x\in C_\ell} \Re \Tr(Q^{\omega_x}_{\ell\setminus e_x}H)$$
  Thus, 
  $$\grad^{euc}_e \Re W_\ell = \sum_{x\in C_\ell} 2Q_{\ell\setminus e_x}^{-\omega_x}$$
  Recalling that the tangent projection is 
  $$P_g X = \frac{1}{2}X - \frac{1}{2}gX^\dagger g -\frac{\eta}{2N} \Im \Tr(g^{-1}X - X^\dagger g^{-1})g$$

  We thus get 
  \begin{align*}
  \grad \Re W_\ell(g) &= \sum_{x\in C_\ell} Q_{\ell\setminus e_x}^{-\omega_x} - g Q_{\ell\setminus e_x}^{\omega_x}g +\eta\frac{2i}{N} \Im \Tr(Q_{\ell\setminus e_x}^{\omega_x}g)g\\
  &=\sum_{x\in C_\ell} Q_{\ell\setminus e_x}^{-\omega_x} - g Q_{\ell\setminus e_x}^{\omega_x}g +\eta\frac{2i\omega_x}{N} \Im W_\ell g
  \end{align*}
\end{proof}

\begin{lemma}\label{gradlemma3}
Let $\ell$ be a loop. 
$$\grad \Im W_\ell =
\sum_{x\in C_\ell} \omega_x i Q_{\ell\setminus e_x}^{-\omega_x} + iQ_e \omega_x Q_{\ell\setminus e_x}^{\omega_x} Q_e- \eta\frac{2i\omega_x}{N} \Re W_\ell Q_e $$
\end{lemma}
\begin{proof}
Recall that $\Im W_\ell = \Re(-iW_\ell)$. 
Note that $$\Im \Tr(A g^\omega B) = \Im \Tr((BA)g^\omega) = \omega \Im \Tr((BA)^\omega g)$$
Thus, 
$$d\Im W_\ell(H) = \sum_{x\in C_\ell}  \omega_x\Re\Tr(-iQ_{\ell\setminus e_x}^{\omega_x}H)$$
and the euclidean gradient is $$\sum_{x\in C_\ell} 2i\omega_x Q^{-\omega_x}_{\ell\setminus e_x}$$
Applying the tangent projection again finally gives for 
$$\grad \Im W_\ell = \sum_{x\in C_\ell} \omega_x i Q_{\ell\setminus e_x}^{-\omega_x} + i g \omega_x Q_{\ell\setminus e_x}^{\omega_x}g - \eta\frac{2i\omega_x}{N} \Re W_\ell g $$

\end{proof}

\begin{lemma}\label{gradiplemma}
Let $\ell_1$ and $\ell_2$ be loops. Then 
$$\ev{\grad W_{\ell_1}, \grad W_{\ell_2}} = \sum_{x\in C_1, y\in C_2, \omega_x \omega_y = -1} 2 W_{\ell_1 \ominus_{x,y}\ell_2} - \sum_{\omega_x \omega_y = 1}2W_{\ell_1\oplus_{x,y}\ell_2} + \eta\frac{2t_1 t_2}{N}W_{\ell_1}W_{\ell_2}$$

\end{lemma}

\begin{proof}
Note that because $\Re \Tr(X)= \Re \Tr(X^\dagger)$, the first part of the gradient has algebra identical to that of the $SO(N)$ case. Thus 
$$\ev{\grad f^R, \grad f'^R} = \sum_{x\in C_1,y\in C_2 }\Re(W_{\ell_1\ominus_{x,y} \ell_2}) - \Re(W_{\ell_1\oplus_{x,y} \ell_2}) + \eta\ev{\frac{2i\omega_x}{N}\Im W_{\ell_1 }g,Q_{\ell_2\setminus e_y}^{-\omega_y} - g Q_{\ell\setminus e_y}^{\omega_y}g}$$
$$ + \eta\ev{Q_{\ell_1\setminus e_x}^{-\omega_x} - gQ_{\ell_1\setminus e_x}^{\omega_x} g, \frac{2i\omega_y}{N} \Im W_{\ell_2}g} + \eta\ev{\frac{2i\omega_x}{N} \Im W_{\ell_1} g, \frac{2i\omega_y}{N}\Im W_{\ell_2}g}$$

The first term after the mergers equals 
$$\eta\frac{\omega_x}{N} \Im W_{\ell_1}\Im \Tr(g^{-1}Q_{\ell_2\setminus e_y}^{-\omega_y} - Q_{\ell_2\setminus e_y}^{\omega_y} g) = -\eta\frac{2\omega_x\omega_y}{N}\Im W_{\ell_1}\Im W_{\ell_2}$$
Similarly, the second term after mergers is 
$$\eta\frac{\omega_y}{N}\Im \Tr(Q_{\ell_2\setminus e_y}^{\omega_y} g)\Re \Tr(i Q_{\ell_1\setminus e_x}^{\omega_x} g - i g^{-1}Q_{\ell_1\setminus e_x}^{-\omega_x}) = -\eta\frac{2\omega_x\omega_y}{N}\Im W_{\ell_1}W_{\ell_2}$$
Finally, the last term is 
$$ \eta\frac{4\omega_x\omega_y}{N^2} \Im W_{\ell_1}\Im W_{\ell_2}\ev{g,g} = \frac{2\omega_x\omega_y}{N} \Im W_{\ell_1}W_{\ell_2}$$
In total, 
$$\ev{\grad f^R, \grad f'^R} = \sum_{x\in C_\ell}\Re(W_{\ell_1\ominus_{x,y} \ell_2}) - \Re(W_{\ell_1\oplus_{x,y} \ell_2}) -\eta\frac{2\omega_x\omega_y}{N}\Im W_{\ell_1}W_{\ell_2}$$

Now for $f^I$. The first term is almost algebraically identical (the $i$s cancel and $\omega$s factor out). But the negative sign is gone. So 

$$\ev{\grad f^I, \grad f'^I} = \sum_{x\in C_1, y\in C_2}\omega_x \omega_y\Re (W_{\ell_1\ominus_{x,y} \ell_2}) + \omega_x\omega_y\Re(W_{\ell_1\oplus_{x,y}\ell_2}) -\eta\frac{2\omega_x\omega_y}{N}\Re W_{\ell_1}\Re W_{\ell_2}$$
$$ -\eta\frac{2\omega_x\omega_y}{N}\Re W_{\ell_2}\Re W_{\ell_2} +\eta\frac{2\omega_x\omega_y}{N}\Re W_{\ell_1}\Re W_{\ell_2}$$
$$ =\sum_{x\in C_1, y\in C_2}\omega_x\omega_y \Re W_{\ell_1\ominus_{x,y} \ell_2} + \omega_x\omega_y\Re W_{\ell_1\oplus_{x,y} \ell_2} - \eta\frac{2\omega_x\omega_y}{N}\Re W_{\ell_1}\Re W_{\ell_2}$$

In total then we have 
$$\Re\ev{\grad W_{\ell_1}, \grad W_{\ell_2}}$$ $$ = \sum_{x\in C_1, y\in C_2}(1-\omega_x\omega_y)\Re(W_{\ell_1\ominus_{x,y} \ell_2}) - (1+\omega_x\omega_y)\Re W_{\ell_1\oplus_{x,y} \ell_2} +\eta\frac{2\omega_x\omega_y}{N}\qty(\Re W_{\ell_1}\Re W_{\ell_2} - \eta\Im W_{\ell_1}\Im W_{\ell_2})$$
$$ = \sum_{x\in C_1, y\in C_2}(1-\omega_x\omega_y)\Re(W_{\ell_1\ominus_{x,y} \ell_2}) - (1+\omega_x\omega_y)\Re W_{\ell_1\oplus_{x,y} \ell_2} +\eta\frac{2\omega_x\omega_y}{N}\Re(W_{\ell_1}W_{\ell_2})$$

Now for the imaginary part. 
$$\ev{Q_{\ell_1\setminus e_x}^{-\omega_x} -gQ_{\ell_1}^{\omega_x} g, i\omega_y Q_{\ell_2\setminus e_y}^{-\omega_y} + i \omega_y g Q_{\ell_2\setminus e_y}^{\omega_y} g}=$$ $$\frac{\omega_y}{2}\Re \text{Tr}(iQ_{\ell_1\setminus e_x}^{\omega_x} Q_{\ell_2\setminus e_Y}^{-\omega_y} + iQ_{\ell_1\setminus e_x}^{\omega_x} gQ_{\ell_2\setminus e_y}^{\omega_y}g-i g^{-1}Q_{\ell_1\setminus e_x}^{-\omega_x}g^{-1} Q_{\ell_2\setminus e_y}^{-\omega_y} $$ $$- i g^{-1}Q_{\ell_1\setminus e_x}^{-\omega_x} Q_{\ell_2\setminus e_y}^{\omega_y}g)$$
$$=-\frac{\omega_y}{2}\Im\Tr(Q_{\ell_1\setminus e_x}^{\omega_x} Q_{\ell_2\setminus e_y}^{-\omega_y} - Q_{\ell_1\setminus e_x}^{-\omega_x} Q_{\ell_2\setminus e_y}^{\omega_y})$$ $$ - \frac{\omega_y}{2} \Im \Tr(Q_{\ell_1\setminus e_x}^{\omega_x} gQ_{\ell_2\setminus e_y}^{\omega_y}g - g^{-1}Q_{\ell_1\setminus e_x}^{-\omega_x}g^{-1}Q_{\ell_2\setminus e_Y}^{-\omega_y})$$
$$  =\omega_y \Im \Tr(Q_{\ell_1\setminus e_x}^{-\omega_x}Q_{\ell_2\setminus e_y}^{\omega_y}) -\omega_y \Im \Tr(Q_{\ell_1\setminus e_x}^{\omega_x} gQ_{\ell_2\setminus e_y}^{\omega_y}g)$$

Recall that in the definition of $\ell_1\oplus\ell_2$ or $\ell_1\ominus \ell_2$, the orientation of the first term does not change. As a result, 
$$ =-\omega_x\omega_y \Im W_{\ell_1\ominus_{x,y} \ell_2} - \omega_x\omega_y \Im W_{\ell_1\oplus_{x,y} \ell_2}$$
Now for the remaining terms,
$$\eta\ev{Q_{\ell_1\setminus e_x}^{-\omega_x} -g Q_{\ell_1\setminus e_x}^{\omega_x} g, -\frac{2i\omega_y}{N}\Re \Tr(Q_{\ell_2\setminus e_y}^{\omega_y}g)} = \eta\frac{\omega_y}{N}\Re\Tr(Q_{\ell_2\setminus e_y}^{\omega_y}g) \Im \Tr(Q_{\ell_1\setminus e_x}^{\omega_x} g - g^{-1}Q_{\ell_1\setminus e_x}^{-\omega_x})$$
$$ =\eta\frac{2\omega_y}{N} \Re \Tr(Q_{\ell_2\setminus e_y}^{\omega_y}g)\Im \Tr(Q_{\ell_1\setminus e_x}^{\omega_x} g) = \eta\frac{2\omega_x\omega_y}{N}\Im W_{\ell_1}\Re W_{\ell_2}$$
Similarly, 
$$\eta\ev{\frac{2i}{N}\Im \Tr(Q_{\ell\setminus e_x}^{\omega_x} g)g, i\omega_xQ_{\ell_2\setminus e_y}^{-\omega_y} +i\omega_x gQ_{\ell_2\setminus e_y}^{\omega_y}g} = \eta\frac{2\omega_y}{N} \Im \Tr(Q_{\ell_1\setminus e_x}^{\omega_x} g)\Re \Tr(Q_{\ell_2\setminus e_y}^{\omega_y}g)$$
$$=\eta\frac{2\omega_x\omega_y}{N} \Im W_{\ell_1}\Re W_{\ell_2}$$
And finally, 
$$\eta\ev{\frac{2i}{N} \Im \Tr(Q_{\ell_1\setminus e_x}^{\omega_x} g)g, -\frac{2i\omega_y}{N}\Re\Tr(Q_{\ell_2\setminus e_y}^{\omega_y}g)g} =- \eta\frac{2\omega_y}{N} \Im \Tr(Q_{\ell_1\setminus e_x}^{\omega_x} g)\Re \Tr(Q_{\ell_2\setminus e_y}^{\omega_y}g)$$
$$=-\frac{2\omega_x\omega_y}{N} \eta\Im W_{\ell_1}\Re W_{\ell_2}$$
In total, 
$$\ev{\grad f^R, \grad f'^I} = \sum_{x\in C_1, y\in C_2}-\omega_x\omega_y \Im W_{\ell_1\ominus_{x,y} \ell_2} - \omega_x\omega_y \Im W_{\ell_1\oplus_{x,y} \ell_2}+\eta\frac{2\omega_x \omega_y}{N} \Im W_{\ell_1}\Re W_{\ell_2}$$
By symmetry, $\ev{\grad f'^R, \grad f^I}$ is the same as $\ev{\grad f^R, \grad f'^I}$ with $\ell_1\to \ell_2$ and $\ell_2\to \ell_1$. So 
$$\ev{\grad f'^R, \grad f^I} =\sum_{x\in C_1, y\in C_2} -\omega_x\omega_y \Im W_{\ell_2\ominus_{x,y} \ell_1} - \omega_x\omega_y \Im W_{\ell_2\oplus_{x,y} \ell_1} + \eta\frac{2\omega_x\omega_y}{N}\Im W_{\ell_2}\Re W_{\ell_1}$$

Now we need to examine the relationship between $W_{\ell_1\oplus \ell_2}$ and $W_{\ell_2\oplus \ell_1}$ and similarly for $\ominus$. For a negative merger, an orientation reversal occurs if $\omega\omega' = 1$. Otherwise it doesn't happen.
Similarly for a positive merger, an orientation reversal happens if $\omega\omega' = -1$. Thus we get 
$$ =\sum_{x\in C_1, y\in C_2}\Im W_{\ell_1\ominus_{x,y} \ell_2} - \Im W_{\ell_1 \oplus_{x,y} \ell_2} +\eta\frac{2\omega_x\omega_y}{N}\Im W_{\ell_2}\Re W_{\ell_1}$$
And so, 
$$\Im \ev{\grad W_{\ell_1}, \grad W_{\ell_2}} =\sum_{x\in C_1, y\in C_2} (1-\omega_x\omega_y)\Im W_{\ell_1\ominus_{x,y} \ell_2} -(\omega_x\omega_y +1)\Im W_{\ell_1\oplus_{x,y}\ell_2} +\eta \frac{2\omega_x\omega_y}{N}(\Im W_{\ell_1}\Re W_{\ell_2} + \Re W_{\ell_1}\Im W_{\ell_2})$$
$$ =\sum_{x\in C_1, y\in C_2}(1-\omega_x\omega_y) \Im W_{\ell_1\ominus_{x,y} \ell_2} - (1+\omega_x \omega_y)\Im W_{\ell_1\oplus_{x,y} \ell_2} +\eta\frac{2\omega_x\omega_y}{N}\Im(W_{\ell_1}W_{\ell_2})$$
Thus we can conclude:
$$\ev{\grad_e W_{\ell_1}, \grad_e W_{\ell_2}} =\sum_{x\in C_1, y\in C_2} (1-\omega_x\omega_y)W_{\ell_1\ominus_{x,y} \ell_2} - (1+\omega_x\omega_y) W_{\ell_1\oplus_{x,y} \ell_2} + \eta\frac{2\omega_x\omega_y}{N} W_{\ell_1}W_{\ell_2}$$
$$= \sum_{\omega_x \omega_y = -1} 2W_{\ell_1\ominus_{x,y} \ell_2} - \sum_{\omega_x \omega_y  =1}2W_{\ell_1\oplus_{x,y} \ell_2} +\eta\sum_{x,y} \frac{2\omega_x\omega_y}{N} W_{\ell_1}W_{\ell_2}$$
Note that
$$\sum_{x\in C_1, y\in C_2}\omega_x \omega_y = \qty(\abs{A_1}\abs{A_2} + \abs{B_1}\abs{B_2}) - \qty(\abs{A_1}\abs{B_2} + \abs{B_1}\abs{A_2})$$
$$=\qty(\abs{A_1}-\abs{B_1})\qty(\abs{A_2}-\abs{B_2}) = t_1 t_2$$
Completing the proof.
\end{proof}

This final lemma is required to account for terms from the measure.
\begin{lemma}\label{gradmeasure}
Let $\ell_1$ and $\ell_2$ be loops. Then 
$$\ev{\grad W_{\ell_1}, \Re \grad W_{\ell_2}} = \sum_{x\in C_1,y\in C_2} W_{\ell_1\ominus_{x,y}\ell_2} - W_{\ell_1\oplus_{x,y} \ell_2} + \eta\frac{t_1 t_2}{N} \qty(W_{\ell_1}W_{\ell_2^{-1}} - W_{\ell_1}W_{\ell_2})$$

\end{lemma}

\begin{proof}
Most of the algebra has already been carried out. 
This inner product is 
$$\ev{\grad \Re W_{\ell_1}, \grad \Im W_{\ell_1}} + i\ev{\grad \Im W_{\ell_1}, \grad \Re W_{\ell_1}}$$
The first term is, as we've computed before, 
$$=\sum_{x\in C_1, y\in C_2}\Re(W_{\ell_1 \ominus_{x,y} \ell_2}) - \Re(W_{\ell_1\oplus_{x,y} \ell_2}) - \eta\frac{2\omega_x\omega_y}{N}\Im W_{\ell_1}\Im W_{\ell_2}$$
The next term is 
$$=\sum_{x\in C_1, y\in C_2} \Im W_{\ell_1\ominus_{x,y} \ell_2} - \Im W_{\ell_1\oplus_{x,y} \ell_2} +\eta\frac{2\omega_x\omega_y}{N}\Re W_{\ell_1}\Im W_{\ell_2}$$
Putting them together, this gives
$$\ev{\grad_e W_\ell, \grad_e \Re W_\ell} = \sum_{x\in C_1,y\in C_2}W_{\ell_1\ominus_{x,y} \ell_2} - W_{\ell_1\oplus-{x,y} \ell_2} - \eta\frac{2\omega_x\omega_y}{N} (i\Re W_{\ell_1} \Im W_{\ell_2} - \Im W_{\ell_1}\Im W_{\ell_2})$$

That last term can be simplified: 
$$= (i\Re W_{\ell_1} - \Im W_{\ell_1}) \Im W_{\ell_2} = i W_{\ell_1} \Im W_{\ell_1} = \frac{1}{2}W_{\ell_1}W_{\ell_2} - \frac{1}{2}W_{\ell_1} W_{-\ell_2}$$
So in total, 
$$\ev{\grad_e W_{\ell_1}, \grad_e \Re W_{\ell_2}} = \sum_{x\in C_1, y\in C_2} W_{\ell_1\ominus_{x,y} \ell_2} - W_{\ell_1 \oplus_{x,y} \ell_2} +\eta\frac{\omega_x\omega_y}{N} W_{\ell_1}W_{-\ell_2} - \eta\frac{\omega_x\omega_y}{N} W_{\ell_1}W_{\ell_2}$$
$$ = \sum_{x\in C_1, y\in C_2} W_{\ell_1\ominus_{x,y} \ell_2} - W_{\ell_1 \oplus_{x,y} \ell_2} +\frac{t_1 t_2}{N} W_{\ell_1}W_{-\ell_2} - \eta\frac{t_1 t_2}{N} W_{\ell_1}W_{\ell_2}$$
\end{proof}

\subsection{Laplacian of Wilson loops}
\begin{lemma}\label{suntrace}
Let $L_X$ and $R_Y$ denote left and right-multiplication, respectively. Then 
$$\Tr(P_g L_X R_Y P_g) = \Re(\Tr X \Tr Y) - \frac{\eta}{N}\Re\Tr(g^{-1}XgY)$$
\end{lemma}

\begin{proof}
We first compute the trace on $U(N)$. An orthonormal basis is $g(e_i e_j^T-e_je_i^T)$ $i<j$, $ig(e_ie_j^T+e_je_i^T)$, $i<j$, and $g\sqrt{2}e_i e_i^T$.
\begin{align*}
&\Tr(P_g L_X L_Y P_g) = \sum_{i<j} \ev{e_i e_j^T - e_j e_i^T, g^{-1}X g(e_i e_j^T - e_j e_i^T)Y}+\sum_{i<j} \ev{e_i e_j^T + e_j e_i^T, g^{-1}X g(e_i e_j^T + e_j e_i^T)Y}  \\&+ \sum_{i}2\ev{e_i e_i^T, g^{-1}X g e_i e_i^T}\\
&= \sum_{i<j}2\qty(\ev{e_i e_j^T, g^{-1}Xg e_i e_j^T Y } + \ev{e_j e_i^T, g^{-1}Xg e_j e_i^T}) + \sum_i 2 \ev{e_i e_i^T, g^{-1}Xg e_i e_i^T}\\
&= \sum_{ij}2\ev{e_i e_j^T, g^{-1}Xg e_i e_j^T} - 2\sum_i \ev{e_i e_i^T, g^{-1}Xg e_i e_i^T} + 2\sum_i \ev{e_i e_i, g^{-1}Xg e_i e_i^T} = \sum_{ij}\Re(g^{-1}Xg)_{ii}Y_{jj}\\
 &= \Re \Tr X \Tr Y
\end{align*}
Now, $g\mathfrak{u}(N) = g\mathfrak{su}(N)\oplus i \mathbb{R}g$. 
Thus the trace on $SU(N)$ is 
$$\Tr_{g\mathfrak{u}(N)}(L_X L_Y) - \ev{\sqrt{\frac{2}{N}}ig, \sqrt{\frac{2}{N}}iXgY} =\Tr_{g\mathfrak{u}(N)}(L_X L_Y) - \frac{1}{N}\Re\Tr(g^{-1}XgY) $$ $$ = \Re \Tr X \Tr Y - \frac{1}{N} \Re\Tr(g^{-1}XgY)$$
\end{proof}

\begin{lemma}\label{sunlaplace}
Let $W_\ell$ be a Wilson loop and $\Delta_e$ the Laplace-Beltrami operator at the edge $e$. Then 
$$\Delta_e W_\ell = -\sum_{x\in C_1, y\in C_2, \omega_x\omega_y = 1}2 W_{\times^1_{x,y}\ell}W_{\times^2_{x,y}\ell} +\sum_{\omega_x\omega_y = -1} 2 W_{\times^1_{x,y}\ell}W_{\times^2_{x,y}\ell} -\qty(2mN -\frac{2\eta t^2}{N})W_\ell$$
\end{lemma}
\begin{proof}
The laplacian does not involve any complex multiplication. We can therefore separately compute $\Delta_e \Re W_{\ell}$ and $\Delta_e \Im W_{\ell}$. 

$$\grad_e \Re W_{\ell} = \sum_{x\in C}\grad^x  \Re W_{\ell}$$
so 
$$\nabla_H^{euc} \grad_e \Re W_{\ell} = \sum_{x,y\in C} d_y \grad^x \Re W_{\ell}(H)$$

Now, recall: 
$$\grad_e \Re W_{\ell} =\sum_{x\in C_\ell} Q_{\ell\setminus e_x}^{-\omega_x} - Q_e Q_{\ell\setminus e_x}^{\omega_x} Q_e +\eta\frac{2i\omega_x}{N}\Im W_{\ell}Q_e$$

Taking the differential and accounting for orientations as in the $SO(N)$ proof, 
\begin{align*}
d\grad_e \Re W_\ell(H) &= \sum_{x,y\in C_\ell, \omega_x\omega_y = -1}Q_{\ell\setminus e_x}^{-\omega_x}(g_y \to H) -\sum_{x\neq y \in C_\ell, \omega_x\omega_y = 1}Q_{\ell\setminus e_x}^{-\omega_x} (g_y \to g^{-1}Hg^{-1}) \\ &- HQ_{\ell\setminus e_x}^{\omega_x} Q_e - Q_e Q^{\omega_x}_{\ell\setminus e_x}H  - \sum_{x\neq y \in C_\ell, \omega_x\omega_y = 1} Q_{\ell\setminus e_x}^{\omega_x}(g_y \to H) \\ &+ \sum_{x,y\in C_\ell, \omega_x\omega_y = -1} Q_e Q_{\ell\setminus e_x}^{\omega_x}(g_y\to g^{-1}H g^{-1}) + \eta\frac{2i\omega_x}{N}\Im W_\ell H +\eta\frac{2i\omega_x}{N} d\Im W_\ell Q_e
\end{align*}
For that last term, 
$$d \Im \Tr W_{\ell}(H)  = \sum_{y\in C}\omega_y\Im \Tr(Q_{\ell\setminus e_y}^{\omega_y}H)$$

Putting it all together, 
$$d\grad_e \Re W_{\ell}(H)= -\sum_{\omega_x\omega_y=1} Q_{\ell\setminus e_x}^{-\omega_x}(Q_y \to Q_e^{-1}HQ_e^{-1}) + \sum_{\omega_x\omega_y = -1} Q_{\ell\setminus e_x}^{-\omega_x}(Q_y\to H)$$
$$ -\sum_{x\in C}H Q_{\ell\setminus e_x}^{\omega_x}Q_e + Q_eQ^{\omega_x}_{\ell\setminus e_x} H  -\sum_{\omega_x\omega_y=1} Q_e Q_{\ell\setminus e_x}^{\omega_x}(Q_y\to H) Q_e$$
$$ +\sum_{\omega_x\omega_y = -1} Q_e Q_{\ell\setminus e_x}^{\omega_x} (Q_e \to Q_e^{-1}H Q_{e}^{-1})Q_e + \eta\sum_{x\in C} \frac{2i\omega_x}{N}H \Im W_{\ell} +\eta\sum_{x\in C, y\in C}\frac{2i\omega_x \omega_y}{N}Q_e  \Im \Tr(Q^{\omega_y}_{\ell\setminus e_y} H)$$
The last two terms can be dropped, as both of their images lie in the normal bundle.

We now compute the trace (recall \cref{suntrace}).
Once again recall that if $Q_{\ell\setminus e_x} = P_+ g^{\omega_y}P_-$, then 
$$Q_{\ell\setminus e_x}^{\pm\omega_x} = P_{\pm\omega_x}^{\pm \omega_x}g^{\pm \omega_x \omega_y} P_{\mp \omega_x}^{\pm \omega_x}$$
First: 

\begin{align*}&\sum_{\omega_x \omega_y = 1}\Tr Q_{\ell\setminus e_x}^{-\omega_x}(Q_y\to Q_e^{-1}HQ_e^{-1} ) = \sum_{\omega_x\omega_y = 1} \Tr P_{-\omega_x}^{-\omega_x} Q_e^{-1}HQ_{e}^{-1} P_{\omega_x}^{-\omega_x} \\ &= \sum_{\omega_x\omega_y = 1} \Re(\Tr(P^{-\omega_x}_{-\omega_x}Q_{e}^{-1})\Tr(Q_e^{-1}P^{-\omega_x}_{\omega_x})) - \frac{\eta}{N}\Re\Tr(g^{-1}(P^{-\omega_x}_{-\omega_x}g^{-1})g g^{-1}P^{-\omega_x}_{\omega_x})\\
&=\sum_{\omega_x\omega_y = 1}\Re(W_{\times^1_{x,y}\ell} W_{\times^2_{x,y}\ell}) - \frac{\eta}{N}\Re \Tr(g^{-1}P^{-\omega_x}_{-\omega_x} g^{-1}P^{-\omega_x}_{\omega_x})=\sum_{\omega_x\omega_y = 1} \Re(W_{\times^1_{x,y}\ell} W_{\times^2_{x,y}\ell}) - \frac{\eta}{N} \Re W_\ell
\end{align*}

Next, we have 
\begin{align*}
&\sum_{\omega_x \omega_y = -1}\Tr Q_{\ell\setminus e_x}^{-\omega_x}(Q_y\to H) = \sum_{\omega_x\omega_y = -1} \Tr(P^{-\omega_x}_{-\omega_x} H P^{-\omega_x}_{\omega_x}) \\ &= \sum_{\omega_x\omega_y = -1}\Re(\Tr(P^{-\omega_x}_{-\omega_x})\Tr(P^{-\omega_x}_{\omega_z})) - \frac{\eta}{N}\Re\Tr(g^{-1} P^{-\omega_x}_{-\omega_x} g P^{-\omega_x}_{\omega_x})\\
&= \sum_{\omega_x\omega_y = -1}\Re(W_{\times^1_{x,y}\ell}W_{\times^2_{x,y}\ell}) - \frac{\eta}{N}\Re W_\ell
\end{align*}

Next, 
\begin{align*}
&\Tr(HQ^{\omega_x}_{\ell\setminus e_x} Q_e + Q_e Q^{\omega_x}_{\ell\setminus e_x}H) =2N\Re \Tr(Q^{-\omega_x}_{\ell\setminus e_x}Q_e) - \frac{2\eta}{N}\Re\Tr(Q^{\omega_x}_{\ell\setminus e_x}Q_e) = \qty(2N-\frac{2\eta}{N})\Re W_\ell
\end{align*}

Next, 
\begin{align*}
&\sum_{\omega_x\omega_y = 1}\Tr Q_e Q^{\omega_x}_{\ell\setminus e_x}(Q_y \to H)Q_e = \Tr(Q_eP^{\omega_x}_{\omega_x} H P^{\omega_x}_{-\omega_x}Q_e)\\
&=\Re\Tr(Q_e P^{\omega_x}_{\omega_x})\Tr(P^{\omega_x}_{\omega_x}Q_e) - \frac{\eta}{N}\Re(Q_e^{-1}Q_e P^{\omega_x}_{\omega_x}Q_e P^{\omega_x}_{-\omega_x}Q_e)\\&= \Re W_{\times^1_{x,y}\ell}W_{\times^2_{x,y}\ell} - \frac{\eta}{N} \Re W_\ell
\end{align*}

Finally, 
\begin{align*}
&\sum_{\omega_x\omega_y = -1} \Tr(Q_e Q_{\ell\setminus e_x}^{\omega_x}(Q_y \to Q_e^{-1}HQ_e^{-1}))= \sum_{\omega_x\omega_y = -1}\Tr(Q_e P^{\omega_x}_{\omega_x} Q_e^{-1} H Q_e^{-1} P^{\omega_x}_{-\omega_x}Q_e)\\
&=\sum_{\omega_x\omega_y = -1} \Re(\Tr(Q_e P^{\omega_x}_{\omega_x}Q_e^{-1})\Tr(Q_e P^{\omega_x}_{-\omega_x}Q_e^{-1})) - \frac{\eta}{N}\Re \Tr(Q_e^{-1} Q_e P^{\omega_x}_{\omega_x} Q_e^{-1}Q_e Q_e^{-1}P^{\omega_x}_{-\omega_x}Q_e)\\
&=\sum_{\omega_x\omega_y = -1}\Re W_{\times^1_{x,y}\ell}W_{\times^2_{x,y}\ell} - \frac{\eta}{N} \Re \Tr(P^{\omega_x}_{\omega_x}Q_e^{-1} P^{\omega_x}_{-\omega_x} Q_e)  = \sum_{\omega_x\omega_y = -1} \Re W_{\times^1_{x,y}\ell}W_{\times^2_{x,y}\ell} - \frac{\eta}{N}\Re W_\ell
\end{align*}

Inserting these identities back in, 
\begin{align*}
\Delta_e W_\ell &= -\qty(\sum_{\omega_x\omega_y = 1} \Re(W_{\times^1_{x,y}\ell} W_{\times^2_{x,y}\ell}) - \frac{\eta}{N} \Re W_\ell) + \qty(\sum_{\omega_x\omega_y = -1}\Re(W_{\times^1_{x,y}\ell}W_{\times^2_{x,y}\ell}) - \frac{\eta}{N}\Re W_\ell)\\
&-\sum_{x}\qty(2N-\frac{2\eta}{N})\Re W_\ell -\sum_{\omega_x\omega_y = 1}\qty(\Re W_{\times^1_{x,y}\ell}W_{\times^2_{x,y}\ell} - \frac{\eta}{N} \Re W_\ell)\\
&+\qty(\sum_{\omega_x\omega_y = -1} \Re W_{\times^1_{x,y}\ell}W_{\times^2_{x,y}\ell} - \frac{\eta}{N}\Re W_\ell)\\ 
&= -\sum_{x\in C}\qty(2N-\frac{2\eta}{N})\Re W_\ell -2\sum_{\omega_x \omega_y = 1}\Re(W_{\times^1_{x,y}\ell}W_{\times^2_{x,y}\ell}) + 2\sum_{\omega_x\omega_y = -1} \Re(W_{\times^1_{x,y}\ell}W_{\times^2_{x,y}\ell}) \\&+ \frac{2\eta}{N}\bigg(\sum_{x\neq y, \omega_x \omega_y = 1}1-\sum_{\omega_x \omega_y = -1}1\bigg)\Re W_\ell
\end{align*}

The last term can be simplified as follows: $\sum_{\omega_x \omega_y = 1, x\neq y} = \abs{A}(\abs{A}-1) + \abs{B}(\abs{B}-1)$ and $\sum_{\omega_x\omega_y = -1} = 2\abs{A}\abs{B}$
So that term reduces to
$$\frac{2\eta}{N}\qty[\abs{A}^2 + \abs{B}^2 - 2\abs{A}\abs{B} -\abs{A}-\abs{B}]\Re W_\ell =\frac{2\eta}{N}(t^2-m) $$ where $t= \abs{A}-\abs{B}$ and $m=\abs{A}+\abs{B}$. $A,B$ are the number of $e$s and $-e$s respectively.
So in total, 
$$\Delta_e\Re  W_{\ell} = -\sum_{\omega_x\omega_y = 1}2\Re W_{\ell\times^1_{x,y}}W_{\ell\times^2_{x,y}} + \sum_{\omega_x \omega_y = -1} 2\Re W_{\ell\times^1_{x,y}}W_{\ell\times^2_{x,y}} - \qty(2mN  - \frac{2\eta t^2}{N})\Re W_{\ell}$$
Now for the imaginary part of the laplacian, notice that no part of the algebra actually used that the matrices between occurrences of the edge $e$ belong to $SU(N)$. We only needed that they are unitary.
Thus, pick any edge that is not $\pm e$, and replace the matrix there with $-iA$. This gives 
$$\Delta_e \Im W_\ell = -2\sum_{\omega_x\omega_y = 1}\Im W_{\ell\times^1_{x,y}}W_{\ell\times^2_{x,y}} + \sum_{\omega_x\omega_y = -1} 2\Im W_{\ell\times^1_{x,y}}W_{\ell\times^2_{x,y}} - \qty(2mN - \frac{2\eta t^2}{N})\Im W_{\ell}$$
And so, we have 
$$\Delta W_{\ell} = -2\sum_{\omega_x\omega_y=1} W_{\ell\times^1_{x,y}}W_{\ell\times^2_{x,y}} + \sum_{\omega_x\omega_y = -1} 2W_{\ell\times^1_{x,y}}W_{\ell\times^2_{x,y}} - \qty(2mN - \frac{2\eta t^2}{N})W_\ell$$

\end{proof}
\subsection{Master loop equation for $U(N)$ and $SU(N)$}
\begin{proof}[Proof of \cref{thm2}]
Let $(\ell_1, \dots, \ell_n)$ be a sequence of loops. By Laplacian integration by parts, 
\begin{align*}
&-\mathbb{E}\qty[(\Delta_e W_{\ell_1})W_{\ell_2}\dots W_{\ell_n}] = Z^{-1}\int_{G^{E^+_\Lambda}} \ev{\grad W_{\ell_1}, \grad \qty(W_{\ell_2}\dots W_{\ell_n}\exp(\beta N \sum_{p\in \mathcal{P}^+_\Lambda} \Re \Tr W_p))}d\mu\\ 
&=\sum_{i=2}^n\mathbb{E}\qty[\ev{\grad W_{\ell_1},\grad W_{\ell_i}}\prod_{j\neq 1, i} W_{\ell_i} ] + \beta N\sum_{i=2}^n \sum_{p\in\mathcal{P}^+(e)}\mathbb{E}\qty[\ev{\grad W_{\ell_1}, \Re \grad W_p} W_{\ell_2}\dots W_{\ell_n}]
\intertext{Applying \cref{gradiplemma} and \cref{gradmeasure},}
&= \sum_{i=2}^n \sum_{x\in C_1, y\in C_i, \omega_x \omega_y = -1}2\mathbb{E}\qty[W_{\ell_1\ominus_{x,y}\ell_i}\prod_{j\neq i, 1}W_{\ell_j}] - \sum_{i=2}^n \sum_{x\in C_1, y\in C_i, \omega_x \omega_y = 1} 2\mathbb{E}\qty[W_{\ell_1\oplus_{x,y}\ell_i}\prod_{j\neq i, 1}W_{\ell_j}]\\
&+\eta \sum_{i=2}^n \frac{2t_1 t_i}{N}\mathbb{E}\qty[W_{\ell_1}\dots W_{\ell_n}]+ \beta N\sum_{i=2}^n \sum_{p\in\mathcal{P}^+(e), x\in C_1}\mathbb{E}\qty[W_{\ell_1\ominus_x p} W_{\ell_2}\dots W_{\ell_n}] 
\\&-\beta N\sum_{i=2}^n \sum_{p\in\mathcal{P}^+(e), x\in C_1}\mathbb{E}\qty[W_{\ell_1\oplus_x p} W_{\ell_2}\dots W_{\ell_n}] +\eta \beta N\sum_{p\in\mathcal{P}_\Lambda^+(e)}\frac{t_1}{N} \mathbb{E}[W_{\ell_1}W_{p^{-1}}W_{\ell_2}\dots W_{\ell_n}]\\  &-\eta \beta N \sum_{p\in\mathcal{P}_\Lambda^+(e)} \frac{t_1 }{N}\mathbb{E}[W_{\ell_1}W_{p}W_{\ell_2}\dots W_{\ell_n}]\\
\intertext{Regrouping the plaquettes in the expansion terms,}
&= \sum_{i=2}^n \sum_{x\in C_1, y\in C_i, \omega_x \omega_y = -1}2\mathbb{E}\qty[W_{\ell_1\ominus_{x,y}\ell_i}\prod_{j\neq i, 1}W_{\ell_j}] - \sum_{i=2}^n \sum_{x\in C_1, y\in C_i, \omega_x \omega_y = 1} 2\mathbb{E}\qty[W_{\ell_1\oplus_{x,y}\ell_i}\prod_{j\neq i, 1}W_{\ell_j}]\\
&+\eta \sum_{i=2}^n \frac{2t_1 t_i}{N}\mathbb{E}\qty[W_{\ell_1}\dots W_{\ell_n}]+ \beta N\sum_{i=2}^n \sum_{p\in\mathcal{P}^+(e), x\in C_1}\mathbb{E}\qty[W_{\ell_1\ominus_x p} W_{\ell_2}\dots W_{\ell_n}] 
\\&-\beta N\sum_{i=2}^n \sum_{p\in\mathcal{P}^+(e), x\in C_1}\mathbb{E}\qty[W_{\ell_1\oplus_x p} W_{\ell_2}\dots W_{\ell_n}] -\eta \beta N\sum_{p\in\mathcal{P}_\Lambda(e)}\frac{t_1t_p}{N} \mathbb{E}[W_{\ell_1}W_{p^{-1}}W_{\ell_2}\dots W_{\ell_n}]\\
\end{align*}
Now for the left hand side, \cref{sunlaplace} gives 
\begin{align*}
&-\mathbb{E}\qty[\qty(\Delta_e W_{\ell_1})W_{\ell_2}\dots W_{\ell_n}] = \sum_{x\neq y\in C_1, \omega_x \omega_y = 1}2\mathbb{E}\qty[W_{\times^1_{x,y}\ell_1}W_{\times^2_{x,y}\ell_1}W_{\ell_2}\dots W_{\ell_n}] \\ &- \sum_{x,y\in C_1, \omega_x \omega_y = -1}2\mathbb{E}\qty[W_{\times^1_{x,y}\ell_1}W_{\times^2_{x,y}\ell_1}W_{\ell_2}\dots W_{\ell_n}] +\qty(2mN - \frac{2\eta t_1^2}{N})\mathbb{E}\qty[W_{\ell_1}\dots W_{\ell_n}]
\end{align*}
Setting the two sides equal,rearranging terms, and dividing by $2$ gives
\begin{align*}
&\qty(mN - \frac{\eta t_1 t}{N})\mathbb{E}\qty[W_{\ell_1}\dots W_{\ell_n}]= \sum_{x,y\in C_1, \omega_x\omega_y = -1}\mathbb{E}\qty[W_{\times^1_{x,y}\ell_1}W_{\times^2_{x,y}\ell_1}W_{\ell_2}\dots W_{\ell)n}] \\ &- \sum_{x\neq y \in C_1, \omega_x \omega_y = 1}\mathbb{E}\qty[W_{\times^1_{x,y}\ell_1}W_{\times^2_{x,y}\ell_1}W_{\ell_2}\dots W_{\ell_n}] +\sum_{i=2}^n \sum_{x\in C_1, y\in C_i, \omega_x \omega_y = -1}\mathbb{E}\qty[W_{\ell_1\ominus_{x,y}\ell_i}\prod_{j\neq i, 1}W_{\ell_j}] \\ &- \sum_{i=2}^n \sum_{x\in C_1, y\in C_i, \omega_x \omega_y = 1} \mathbb{E}\qty[W_{\ell_1\oplus_{x,y}\ell_i}\prod_{j\neq i, 1}W_{\ell_j}]
+ \frac{\beta N}{2}\sum_{i=2}^n \sum_{p\in\mathcal{P}^+(e), x\in C_1}\mathbb{E}\qty[W_{\ell_1\ominus_x p} W_{\ell_2}\dots W_{\ell_n}] 
\\&-\frac{\beta N}{2}\sum_{i=2}^n \sum_{p\in\mathcal{P}^+(e), x\in C_1}\mathbb{E}\qty[W_{\ell_1\oplus_x p} W_{\ell_2}\dots W_{\ell_n}] -\eta \beta \sum_{p\in\mathcal{P}_\Lambda(e)}t_1t_p \mathbb{E}[W_{\ell_1}W_{p^{-1}}W_{\ell_2}\dots W_{\ell_n}]\\
\end{align*}
\end{proof}

\bibliographystyle{alpha}
\bibliography{MasterLoop}

\Addresses

\newpage

\appendix
\section{Comments on Other Approaches}

In this appendix we compare our approach to the approaches of \cite{Ch19a} and \cite{SheSmZh22}.

\subsection{Integration by parts and exchangeable pairs}\label{IBP section}

In this section we first show that the Schwinger Dyson equation in \cite{Ch19a} is simply integration by parts on $SO(N)$ written in extrinsic coordinates by obtaining integration by parts from Stein's method of exchangeable pairs. We then generalize this to all compact Lie groups with Riemannian structure inducing the Haar measure.

Recall
\begin{definition}
    A pair of random variables $(U,U')$ is called an \textit{exchangeable pair} if $(U,U')=(U',U)$ in distribution.
\end{definition}

We have the following elementary lemma for exchangeable pairs,
\begin{lemma}
    Let $(U,U')$ be a exchangeable pair and $f,g$ real valued Borel measurable functions. Then
    \begin{equation}
        \mathbb{E}[(f(U')-f(U))g(U)]=-\frac{1}{2}\mathbb{E}[(f(U')-f(U))(g(U')-g(U))]
    \end{equation}
    \label{Stein Lemma}
\end{lemma}

\begin{proof}
    See Lemma 6.1 from \cite{Ch19a}
\end{proof}


 Recall that $SO(N)$ has a Riemannian metric given by the inner product $\langle A,B \rangle=\frac{1}{2}\mathrm{Tr}(A^T B)$ on $\mathfrak{so}(N)$. Thus $\mathfrak{so}(N)$ has an orthonormal basis $=e_{i,j}-e_{j,i}=\frac{d}{d\epsilon}|_{\epsilon=0} (\sqrt{1-\epsilon^2}e_{i,i}+\epsilon e_{i,j}-\epsilon e_{j,i}+\sqrt{1-\epsilon^2})e_{j,j}$ for $1 \leq i<j \leq N$. These basis vectors are simply the tangent vectors generated by rotations in the plane spanned by two basis vectors.\\

So working with the rotations $R_{i,j}(\epsilon):=(\sqrt{1-\epsilon^2}e_{i,i}+\epsilon e_{i,j}-\epsilon e_{j,i}+\sqrt{1-\epsilon^2}e_{j,j}$ at the identity, it follows that the Laplace Beltrami operator of of $f \in C^2(SO(N))$ is simply given by $\Delta f(O) = \sum_{1 \leq i<j \leq N} \frac{d^2}{d\epsilon^2}|_{\epsilon=0} f(R_{i,j}(\epsilon) O)$

So
\begin{equation}
    \begin{split}
        &\Delta f= \sum_{1\leq i <j \leq N}\lim_{\epsilon \downarrow 0} \frac{f(R_{i,j}(\epsilon) O)+f(R_{i,j}(-\epsilon) O)-2f(O)}{\epsilon^2}\\
        &= 2N(N-1)\lim_{\epsilon \downarrow 0} \frac{1}{\epsilon^2}\mathbb{E}_{(I,J),\eta}[ f(R_{I,J}(\eta \epsilon) O)-f(O)]
    \end{split}
\end{equation}
where the expectation is taken over $\eta$ uniformly chosen from $\{\pm 1\}$ and $(I,J)$ uniformly from $\{1\leq i < j \leq N\}$.
Hence if $g:SO(N) \to \mathbb{R}$ is another smooth function we have
\begin{equation}
    \begin{split}
        \int_{SO(N)} \Delta f g  = 2N(N-1)\lim_{\epsilon \downarrow 0} \frac{1}{\epsilon^2}\mathbb{E}[ (f(R_{I,J}(\eta \epsilon) O)-f(O)) g(O)]
    \end{split}
    \label{Laplacian term}
\end{equation}
where the expectation is now taken with respect to the Haar measure on $SO(N)$, $(I,J)$ and $\eta$. Since $SO(N)$ is compact and $f,g \in C^2$, the interchange of limit and expectation above follows from the bounded convergence theorem.\\

Similarly for $f, g \in C^2(SO(N))$
\begin{equation}
    \begin{split}
        \int_{SO(N)} \langle \nabla f, \nabla g \rangle &= \int_{SO(N)}\sum_{1\leq i<j\leq N} \frac{d}{d\epsilon}|_{\epsilon=0} f(R_{i,j}(\epsilon)O)\frac{d}{d\epsilon}|_{\epsilon=0}  g(R_{i,j}(\epsilon)O) dO\\
       &=\lim_{\epsilon \downarrow 0} \frac{1}{\epsilon^2}\int_{SO(N)}\sum_{1\leq i<j\leq N} \frac{1}{2}[(f(R_{i,j}(\epsilon)O)-f(O))(g(R_{i,j}(\epsilon)O)-g(O))\\
       &+(f(R_{i,j}(-\epsilon)O)-f(O))(g(R_{i,j}(-\epsilon)O)-g(O))]dO\\
       &=  N(N-1)\lim_{\epsilon \downarrow 0} \frac{1}{\epsilon^2}\mathbb{E}[(f(R_{I,J}(\eta \epsilon)O)-f(O))(g(R_{I,J}(\eta \epsilon)O)-g(O))]
    \end{split}
    \label{grad inner product}
\end{equation}

Finally in \cite{Ch19a}, it is shown that $(O,R_{i,j}(\eta \epsilon)O)$ form an exchangeable pair, thus combining \eqref{Laplacian term} and \eqref{grad inner product}, integration by parts on $SO(N)$, $\int_{SO(N)} \Delta f g=-\int_{SO(N)} \langle \nabla f, \nabla g \rangle$
 follows as consequence of lemma \ref{Stein Lemma}. Moreover since lemma \ref{Stein Lemma} with this Stein pair is Chatterjee's starting point for deriving his Schwinger-Dyson equation (Theorem 7.1 \cite{Ch19a}), we see that it is simply integration by parts written in extrinsic coordinates.\\

This derivation of integration by parts actually generalizes to any compact Lie group with Riemannian metric inducing the Haar measure. For such a Lie group $G$, let $\{e_i\}_{i=1}^{d}$ be an orthonormal basis for the corresponding Lie algebra $\mathfrak{g}$. We can write $e_i = \frac{d}{dt}|_{t=0} g_i(t)$ where $g_i(t):=\exp(te_i)\in G$. Now let $g$ be an element of $G$ selected from the Haar measure, and $g_{\epsilon}:=  g_{u([d])}(\eta \epsilon) g$ where $\eta$ is Bernoulli with $\mathbb{P}(\eta=1)=\mathbb{P}(\eta=-1)=\frac{1}{2}$ and $u([d])$ is uniform on $[d]=\{1,...,d\}$ all chose independently of each other.

\begin{lemma}
    $(g,g_{\epsilon})$ is an exchangeable pair
\end{lemma}
\begin{proof}
    By left invariance of the Haar measure, $g_{\epsilon}$ is also distributed according to the Haar measure. Moreover $g=g_{u([d])}(-\eta \epsilon) g_{\epsilon}$
    so $(g,g_{\epsilon})=(g_{\epsilon},g)$ in distribution since $-\eta=\eta$ in distribution.
\end{proof}

Now applying Lemma $\ref{Stein Lemma}$ to this Stein pair with $f,g \in C^2(G)$, we obtain integration by parts on $G$ by an identical calculation to that in equation \eqref{grad inner product} and \eqref{Laplacian term} for the $SO(N)$ case.

\subsection{Symmetrized Master Loop Equation and Langevin Dynamics}\label{Langevin section}

One immediate corollary of Theorem \ref{thm1} and Theorem \ref{thm2} is the symmetrized master loop equation (Theorem 1 \cite{SheSmZh22}).

This corollary is derived in \cite{SheSmZh22} through studying the following Langevin dynamics with invariant Yang-Mill's measure.

\begin{equation}
\label{Langevin}
    dQ = \frac{1}{2}\nabla S(Q)dt+d\mathfrak{B}
\end{equation}
where $S(Q):= N\beta \sum_{p \in \mathcal{P}} \mathrm{Re}(\mathrm{Tr}(Q_p))$ is the Yang-Mills action, and $\nabla$ is the intrinsic gradient on $G^{\Lambda}$. 

It can then be shown via integration by parts that the Yang-Mills measure is invariant under the dynamics \eqref{Langevin}, (Lemma 3.3 \cite{SheSmZh22}).\\

The rest of the proof proceeds by applying Itô's formula in $(\mathbb{R}^{N^2})^{\Lambda}$ to $f(Q):= W_{\ell_1}...W_{\ell_n}$ with $Q$ evolving according to $\eqref{Langevin}$ starting with the Yang-Mills measure as the initial distribution, and taking the expectation of both sides with respect the Yang-Mills measure. \\

$f(Q_t)$ is of course a semimartingale \cite{KaSh99}, and can thus be written in the form $f(Q_t)=f(Q_0)+M_t+A_t$ where $M_t$ is a martingale, and $A_t$ is a bounded variation process with $M_0=A_0=0$. The equation \cite{SheSmZh22} obtain is $\mathbb{E}_{\mathrm{YM}}[A_t]=0$. But standard theory for stochastic analysis on manifolds (Chapter 3, \cite{Hsu06}) tells us that the Langevin dynamics $\eqref{Langevin}$ has infinitesimal generator $\mathcal{L}=\frac{1}{2}\Delta + \frac{1}{2}\langle \nabla S , \nabla \rangle $, so $A_t = \int_{0}^{t} \mathcal{L} f(Q_s) ds $. Thus the symmetrized master loop equation derived in \cite{SheSmZh22} is equivalent to the following integration by parts

\begin{equation}
   \int_{G^{\Lambda}}  \Delta f(Q)\exp(S(Q)) + \langle \nabla f(Q), \nabla \exp(S(Q)) \rangle dQ=0
\end{equation}

where $\nabla$ and $\Delta$ are the intrinsic gradient and Laplace-Beltrami operator on $G^{\Lambda}$ respectively.

\section{Deriving the extrinsic integration by parts formula} \label{extrinsic derivation}
Again let $SO(N)\subset\mathbb{R}^{N^2}$ in the natural way.
Let $U\supset SO(N)$ be open, and $f,g$ smooth functions on $U$. Then we can extend the integration by parts formula to functions $f,g\in C^2(U)$.
Let $X_q = q$ be a vector field on $\mathbb{R}^{N^2}$. Note that for consistency we're equipping $\mathbb{R}^{N^2}$ with the metric $\ev{X,Y} = \frac{1}{2}\Tr(X^TY)$. Thus, 
the gradient for this metric is twice the usual gradient.
\newline 

Recall that
$$X_{ij}(q) = qe_i \wedge e_j^T = \sum_{a}q_{ai}e_{a}e_j^T - q_{aj}e_a e_i^T$$
so (because $X_{ij}$ is an orthonormal frame on $SO(N)$)
$$X_{ij}(f) = \sum_{a} q_{ai}\pdv{f}{q_{aj}} - q_{aj}\pdv{f}{q_{ai}}$$
And so $$\ev{\grad_G f, \grad_G g} = \sum_{a,b,i<j}\qty(q_{ai}\pdv{f}{q_{aj}} - q_{aj}\pdv{f}{q_{ai}})\qty(q_{bi}\pdv{g}{q_{bj}} - q_{bj}\pdv{g}{q_{bi}})$$
Because this quantity is symmetric in $i$ and $j$, 
$$ =\frac{1}{2}\sum_{a,b, i\neq j} \qty(q_{ai}\pdv{f}{q_{aj}} - q_{aj}\pdv{f}{q_{ai}})\qty(q_{bi}\pdv{g}{q_{bj}} - q_{bj}\pdv{g}{q_{bi}})$$

$$ =\frac{1}{2}\sum_{a,b, i\neq j} \qty(q_{ai}\pdv{f}{q_{aj}} - q_{aj}\pdv{f}{q_{ai}})\qty(q_{bi}\pdv{g}{q_{bj}} - q_{bj}\pdv{g}{q_{bi}})$$
$$ =\sum_{a,b,i\neq j} q_{ai}q_{bi}\pdv{f}{q_{aj}}\pdv{g}{q_{bj}} - q_{ai} q_{bj}\pdv{f}{q_{aj}}\pdv{g}{q_{bi}}$$
We can simplify the form of the sums by noting that $i=j$ terms vanish anyway:
$$ =\sum_{a,b,i,j} q_{ai}q_{bi}\pdv{f}{q_{aj}}\pdv{g}{q_{bj}} - q_{ai} q_{bj}\pdv{f}{q_{aj}}\pdv{g}{q_{bi}}$$

Then, by orthogonality the first term simplifies further 
$$\sum_{a,b,j}\delta_{ab}\pdv{f}{q_{aj}}\pdv{g}{q_{bj}} - \sum_{a,b,i,j} q_{ai}q_{bj} \pdv{f}{q_{aj}}\pdv{g}{q_{bi}}$$
$$ =\sum_{a,j}\pdv{f}{q_{aj}}\pdv{g}{q_{aj}} - \sum_{a,b,i,j}q_{ai}q_{bj}\pdv{f}{q_{aj}}\pdv{g}{q_{bi}}$$

Next, 
$$\Delta_G f = \sum_{i<j} X_{ij}(\sum_{a} q_{ai} \pdv{f}{q_{aj}} - q_{aj}\pdv{f}{q_{ai}})$$

$$X_{ij}\qty(q_{ai}\pdv{f}{q_{aj}}) = \sum_{b} \qty(q_{bi}\pdv{\qty(q_{ai} \pdv{f}{q_{aj}})}{q_{bj}} - q_{bj}\pdv{\qty(q_{ai}\pdv{f}{q_{aj}})}{q_{bi}})$$
$$ =\sum_b \qty(q_{bi} \delta_{ab}\delta_{ij} \pdv{f}{q_{aj}} + q_{bi}q_{ai} \pdv[2]{f}{q_{bj}}{q_{aj}} - q_{bj}\delta_{ab} \pdv{f}{q_{aj}} - q_{bj}q_{ai}\pdv[2]{f}{q_{bi}}{q_{aj}})$$
$$ = q_{ai}\delta_{ij} \pdv{f}{q_{aj}} - q_{aj}\pdv{f}{q_{aj}} + \sum_{b}\qty(q_{bi}q_{ai}\pdv[2]{f}{q_{bj}}{q_{aj}} - q_{bj}q_{ai}\pdv[2]{f}{q_{bi}}{q_{aj}})$$
Noting that $i<j$
$$= - q_{aj}\pdv{f}{q_{aj}} + \sum_{b}\qty(q_{bi}q_{ai}\pdv[2]{f}{q_{bj}}{q_{aj}} - q_{bj}q_{ai}\pdv[2]{f}{q_{bi}}{q_{aj}})$$
Similarly, 
$$X_{ij}\qty(q_{aj}\pdv{f}{q_{ai}}) = \sum_{b} \qty(q_{bi}\pdv{\qty(q_{aj} \pdv{f}{q_{ai}})}{q_{bj}} - q_{bj}\pdv{\qty(q_{aj}\pdv{f}{q_{ai}})}{q_{bi}})$$
$$ = \sum_{b}\qty(q_{bi}\delta_{ab} \pdv{f}{q_{ai}} + q_{bi}q_{aj}\pdv[2]{f}{q_{bj}}{q_{ai}} - q_{bj}\delta_{ab}\delta_{ij} \pdv{f}{q_{ai}} - q_{bj}q_{aj}\pdv[2]{f}{q_{bi}}{q_{ai}})$$
$$=q_{ai}\pdv{f}{q_{ai}} -q_{aj}\delta_{ij} \pdv{f}{q_{ai}} + \sum_{b} \qty(q_{bi}q_{aj}\pdv[2]{f}{q_{bj}}{q_{ai}} - q_{bj}q_{aj}\pdv[2]{f}{q_{bi}}{q_{ai}})$$
Again using $i<j$
$$=q_{ai}\pdv{f}{q_{ai}}  + \sum_{b} \qty(q_{bi}q_{aj}\pdv[2]{f}{q_{bj}}{q_{ai}} - q_{bj}q_{aj}\pdv[2]{f}{q_{bi}}{q_{ai}})$$
Subtracting the two, 
$$\Delta_G f = -\sum_{a, i<j}\qty(q_{aj}\pdv{f}{q_{aj}} +q_{ai}\pdv{f}{q_{ai}}) $$ $$+ \sum_{a,b,i<j}\qty[\qty(q_{bi}q_{ai}\pdv[2]{f}{q_{bj}}{q_{aj}} - q_{bj}q_{ai}\pdv[2]{f}{q_{bi}}{q_{aj}}) - \qty(q_{bi}q_{aj}\pdv[2]{f}{q_{bj}}{q_{ai}} - q_{bj}q_{aj}\pdv[2]{f}{q_{bi}}{q_{ai}})]$$

For the first term, by the symmetry of the summand in $i,j$, 
$$ =-\frac{1}{2}\sum_{a, i\neq j}\qty(q_{aj}\pdv{f}{q_{aj}} + q_{ai}\pdv{f}{q_{ai}}) = -\sum_{a,i\neq j} q_{aj}\pdv{f}{q_{aj}} = -(N-1)\sum_{a, j}q_{aj}\pdv{f}{q_{aj}}$$
For the second term, it can be regrouped into 
$$\sum_{i<j} \sum_{a,b} \qty[\qty(q_{bi}q_{ai}\pdv[2]{f}{q_{bj}}{q_{aj}} + q_{bj} q_{aj}\pdv[2]{f}{q_{bi}}{q_{ai}}) - \qty(q_{bj}q_{ai}\pdv[2]{f}{q_{bi}}{q_{aj}} + q_{bi}q_{aj}\pdv[2]{f}{q_{bj}}{q_{ai}})]$$
$$ = \sum_{a,b, i\neq j}\qty(q_{bi}q_{ai}\pdv[2]{f}{q_{bj}}{q_{aj}} - q_{bj}q_{ai}\pdv[2]{f}{q_{bi}}{q_{aj}})$$
For the first term,
$$ = \sum_{a,b}   \qty(\sum_{i}q_{bi}q_{ai} \sum_{j\neq i} \pdv[2]{f}{q_{bj}}{q_{aj}})$$
This simplifies to 
$$\sum_{a,b} \sum_{i}q_{bi}q_{ai} \sum_j \pdv[2]{f}{q_{bj}}{q_{aj}} - \sum_{a,b} \sum_i q_{ai}q_{bi}\pdv[2]{f}{q_{bi}}{q_{ai}}$$
By orthogonality, 
$$ = \sum_{a, b}\delta_{ab}\sum_j \pdv[2]{f}{q_{bj}}{q_{aj}} - \sum_{a,b,i}\pdv[2]{f}{q_{bi}}{q_{ai}}$$
$$ = \sum_{a,j}\pdv[2]{f}{q_{aj}} - \sum_{a,b,i}q_{ai}q_{bi}\pdv[2]{f}{q_{bi}}{q_{ai}}$$
What remains is the term 
$$ -\sum_{a,b,i\neq j} q_{bj}q_{ai}\pdv[2]{f}{q_{bi}}{q_{aj}} = -\sum_{a,b,i,j} q_{bj}q_{ai}\pdv[2]{f}{q_{bi}}{q_{aj}} + \sum_{i,a,b}q_{bi}q_{ai} \pdv[2]{f}{q_{bi}}{q_{ai}}$$
This is quite nice because we can see that the unfamiliar term cancels. In summary then, 
$$\Delta_Gf = -(N-1)\sum_{a,j}q_{aj}\pdv{f}{q_{aj}}+\sum_{a,j}\pdv[2]{f}{q_{aj}} - \sum_{a,b,i,j}q_{a,b,i,j}\pdv[2]{f}{q_{bi}}{q_{aj}}$$
Now, recalling the integration by parts on $G$:
$$\int g\Delta_G f = -\int \ev{\grad_G f, \grad_G g}$$
On the LHS then, we have 
$$\int_G -(N-1)\sum_{a,j}q_{aj}g\pdv{f}{q_{aj}}+\sum_{a,j}g\pdv[2]{f}{q_{aj}} - \sum_{a,b,i,j}q_{a,b,i,j}g\pdv[2]{f}{q_{bi}}{q_{aj}}$$
And on the RHS we have 
$$ = -\int_{G} \ev{\grad_G f, \grad_G g} = -\int_G \sum_{a,j}\pdv{f}{q_{aj}}\pdv{g}{q_{aj}} - \sum_{a,b,i,j}q_{ai}q_{bj}\pdv{f}{q_{aj}}\pdv{g}{q_{bi}}$$
Rearranging,
$$(N-1)\int g\sum_{a,j}q_{aj} \pdv{f}{q_{aj}} d\mu$$ $$ = \int_G \qty[\sum_{a,j} g\pdv[2]{f}{q_{aj}} - \sum_{a,b,i,j}q_{bj}q_{ai}\pdv[2]{f}{q_{bi}}{q_{aj}} + \sum_{a,j}\pdv{f}{q_{aj}}\pdv{g}{q_{aj}} - \sum_{a,b,i,j}q_{ai}q_{bj}\pdv{f}{q_{aj}}\pdv{g}{q_{bi}}]d\mu$$

Which is precisely the Schwinger dyson equation.

\end{document}